\documentclass[12pt, a4paper]{article}

\usepackage{amsfonts}
\usepackage{amsmath}
\usepackage{amssymb}
\usepackage[round]{natbib}
\usepackage[dvips]{epsfig}
\usepackage{dcolumn}
\usepackage{enumerate}
\usepackage{hhline}
\usepackage{dsfont}
\usepackage{textcomp}
\usepackage{lineno}
\usepackage[nolists, nomarkers]{endfloat}
\usepackage{graphicx}
\usepackage{color}
\usepackage[usenames,dvipsnames]{xcolor}
\usepackage{accents}
\usepackage{verbatim}

\usepackage{rotating}
\usepackage{hyperref}
\usepackage{multirow}
\usepackage{shortvrb}
\usepackage{xr}
\usepackage{rotating}
\usepackage{paralist}
\usepackage{nicefrac}
\usepackage{floatrow}
\usepackage{float}
\usepackage[percent]{overpic}
\usepackage{subfig}
\usepackage{caption}
\usepackage[font=small]{caption}
\usepackage{scalerel,stackengine}

\usepackage[ruled,vlined]{algorithm2e}
\SetAlCapNameFnt{\small}
\SetAlCapFnt{\small}
\SetAlFnt{\small}

\usepackage{algorithmic}
\algsetup{linenosize=\tiny}
\stackMath
\newcommand\reallywidehat[1]{%
\savestack{\tmpbox}{\stretchto{%
  \scaleto{%
    \scalerel*[\widthof{\ensuremath{#1}}]{\kern-.6pt\bigwedge\kern-.6pt}%
    {\rule[-\textheight/2]{1ex}{\textheight}}
  }{\textheight}%
}{0.5ex}}%
\stackon[1pt]{#1}{\tmpbox}%
}

\usepackage{graphicx}

\usepackage{environ}
\makeatletter
\NewEnviron{Lalign}{\tagsleft@true\begin{flalign}\hspace{3em}\BODY\end{flalign}}
\makeatother

\usepackage{titlesec}

\usepackage[margin=1in]{geometry}
\titlespacing*\section{0pt}{0pt plus 4pt minus 2pt}{0pt plus 2pt minus 2pt}
\titlespacing*\subsection{0pt}{0pt plus 4pt minus 2pt}{0pt plus 2pt minus 2pt}
\titlespacing*\subsubsection{0pt}{0pt plus 4pt minus 2pt}{0pt plus 2pt minus 2pt}

\usepackage{paralist}

\renewenvironment{itemize}[1]{\begin{compactitem}#1}{\end{compactitem}}
\renewenvironment{enumerate}[1]{\begin{compactenum}#1}{\end{compactenum}}

\newcounter{myremark}

\newcommand{\NK}[1]{#1}

 \MakeShortVerb{\°}

\def \dsP {\text{$\mathds{P}$}}
\def \dsE {\text{$\mathds{E}$}}
\def \dsR {\text{$\mathds{R}$}}

\newcommand{\rZ}{Z}
\newcommand{\rY}{Y}
\newcommand{\rX}{\mX}

\newcommand{\ry}{y}
\newcommand{\rx}{\xvec}

\newcommand{\pZ}{F_\rZ}

\newcommand{\dZ}{f_\rZ}

\newcommand{\h}{h}

\newcommand{\basisy}{\avec}

\newcommand{\basisx}{\bvec}
\newcommand{\basisyx}{\cvec}

\newcommand{\parm}{\varthetavec}

\newcommand{\dimparm}{P}

\newcommand{\shiftparm}{\betavec}

\newcommand{\ie}{\textit{i.e.}~}
\newcommand{\eg}{\textit{e.g.}~}

\newcommand{\Prob}{\mathbb{P}}

\newcommand{\RR}{\mathbb{R}}

\usepackage{dsfont}

\newcommand{\ubar}[1]{\underaccent{\bar}{#1}}

 \DeclareMathOperator{\diag}{diag}

 \DeclareMathOperator*{\argmax}{{arg\,max}}

 \DeclareMathOperator{\ND}{N}
 \DeclareMathOperator{\UD}{U}

 \def \calC {\mathcal C}

 \def \calI {\mathcal I}

\def \avec {\text{\boldmath$a$}}    
\def \bvec {\text{\boldmath$b$}}    
\def \cvec {\text{\boldmath$c$}}    
    \def \mD {\text{\boldmath$D$}}

    \def \mI {\text{\boldmath$I$}}

    \def \mP {\text{\boldmath$P$}}
    
    \def \mR {\text{\boldmath$R$}}
    \def \mS {\text{\boldmath$S$}}

\def \xvec {\text{\boldmath$x$}}    \def \mX {\text{\boldmath$X$}}
\def \yvec {\text{\boldmath$y$}}    \def \mY {\text{\boldmath$Y$}}
\def \zvec {\text{\boldmath$z$}}    \def \mZ {\text{\boldmath$Z$}}

\def \ztildevec {\text{\boldmath$\tilde z$}}    \def \mtildeZ {\text{\boldmath$\tilde Z$}}

\def \betavec         {\text{\boldmath$\beta$}}
\def \gammavec        {\text{\boldmath$\gamma$}}

\def \thetavec        {\text{\boldmath$\theta$}}
\def \varthetavec     {\text{\boldmath$\vartheta$}}

\def \lambdavec       {\text{\boldmath$\lambda$}}
\def \muvec           {\text{\boldmath$\mu$}}

\def \mTheta   {\mathbf{\Theta}}
\def \mLambda  {\mathbf{\Lambda}}

\def \mSigma   {\mathbf{\Sigma}}

\def \nullvec {\mathbf{0}}

\usepackage{color}
\usepackage{colordvi}
\newlength{\breite}
\breite\textwidth
\newcounter{aufg}[section]
  {\refstepcounter{aufg}\noindent\textbf{Exercise \arabic{aufg}:}
   \\*[1ex]\noindent}{\vspace{.5cm}}
   
 \newcounter{notes}[section]
  {\refstepcounter{aufg}\noindent\textbf{}
   \\*[1ex]\noindent}{\vspace{.5cm}}
   
\usepackage{amsthm}  
\newtheorem{defin}{Definition} 

\newtheorem{cor}{Corollary}

\theoremstyle{definition}
\newtheorem{remark}[]{Remark}

\newtheorem*{beisp*}{Example}
\newtheorem{Proof}{Proof}


\newtheoremstyle{break}
  {}
  {}
  {}
  {}
  {\bfseries}
  {.}
  {\newline}
  {}
  
\theoremstyle{break}

\newcommand{\head}[2]%
 {\hrule \vspace{.15cm} {\sfbold Advanced Statistical Inference, Summer Term 2012, Georg-August-University G\"ottingen}\hfill
{\sfbold Sheet #1}\\
{\sfbold Prof. Dr. Thomas Kneib, Nadja Klein}\hfill {\sfbold #2}

\vspace{.2cm}
\hrule

\vspace{1cm}

}

\newcounter{auf}
{\refstepcounter{auf}
\begin{center}
\fcolorbox[gray]{0}{.95}{
\makebox[\breite]{
\textbf{Exercise \arabic{auf}}
}}\\*[1ex]\noindent
\end{center}
}{\vspace{.5cm}}

\newcounter{loes}[section]
{\stepcounter{loes}
\begin{center}
\fcolorbox[gray]{0}{.95}{
\makebox[\breite]{
\textbf{L"osung \arabic{loes}}
}}\\*[1ex]\noindent
\end{center}
}{}

{\begin{center}
\fcolorbox[gray]{0}{.95}{
\makebox[\breite]{
\textbf{Zu Aufgabe #1}
}}\\*[1ex]\noindent
\end{center}\vspace{1cm}
}{\vspace{1cm}}

\newcounter{ka}
{\refstepcounter{ka}
\begin{center}
\framebox[\textwidth]{
\textbf{Aufgabe \arabic{ka}} \hfill #1 Punkte
}\\*[1ex]\noindent
\end{center}
}{\vspace{1cm}}

\newcounter{lka}
{\refstepcounter{lka}
\begin{center}
\framebox[\textwidth]{
\textbf{L\"osung \arabic{lka}} \hfill #1 Punkte
}\\*[1ex]\noindent
\end{center}
}{\vspace{1cm}}

\begin{document}
\setdefaultleftmargin{3.5mm}{3mm}{3mm}{3mm}{3mm}{3mm}
\title{Multivariate Conditional Transformation Models}
\author{Nadja Klein$^{1\mbox{}^\star}$, Torsten Hothorn$^2$, Luisa Barbanti$^2$ and Thomas Kneib$^3$\\
 \normalsize $^1$Humboldt-Universit\"at zu Berlin, $^2$Universit\"at Z\"urich, \\\vspace{-0.5em} \normalsize$^3$Georg-August-Universit\"{a}t G\"{o}ttingen
 }
\date{}
\maketitle

\begin{abstract}\footnotesize
\noindent Regression models describing the joint distribution of
multivariate response variables conditional on covariate information have
become an important aspect of contemporary regression analysis.  However, a
limitation of such models is that they often rely on rather simplistic
assumptions, \eg~a constant dependency structure that is not allowed to
vary with the covariates or the restriction to linear dependence between the
responses only.  We propose a general framework for multivariate conditional
transformation models that overcomes these limitations and describes the
entire distribution in a tractable and interpretable yet flexible way
conditional on nonlinear effects of covariates.  The framework can be
embedded into likelihood-based inference, including results on asymptotic
normality, and allows the dependence structure to vary with covariates.
In addition, the framework scales well beyond bivariate response situations,
which were the main focus of most earlier investigations.  We illustrate the
application of multivariate conditional transformation models in a
trivariate analysis of childhood undernutrition and demonstrate empirically
that our approach can be beneficial compared to existing benchmarks such that
complex truly multivariate data-generating processes can be inferred from
observations.
\end{abstract}

\textit{Key words: Constrained optimization; copula; marginal distributions; multivariate regression; most likely transformations; normalizing flows; seemingly unrelated regression.}

\vfill
\noindent
{\small $\mbox{}^\star$  Correspondence should be directed to~Prof.~Dr.~Nadja Klein at Humboldt Universit\"at zu Berlin,
Unter den Linden 6, 10099 Berlin. Email: nadja.klein@hu-berlin.de.}

\section{Introduction}

\NK{In a broad sense, regression models describe the distribution of a
response conditional on a set of covariates.  Such models are a versatile
tool to understand how changes in the covariates propagate to changes in the
distribution of the response.  Distributional and multivariate regression
models have received much interest during the last decade.  Rather than
focusing on the conditional mean, distributional regression strives to
describe relevant features of the complete conditional distribution of a
usually univariate response by flexible functions of the covariates.
Multivariate regression models employ covariates to express the joint
conditional distribution of a multivariate response.  Known for half a
century, transformation models have recently received renewed interest in
statistics as an important technique for distributional regression and,
under the term normalizing flows, in machine
learning~\citep{papamakarios2019normalizing} for modelling high-dimensional
responses in an unconditional way.  The core idea of flows or
transformation models is to apply a data-driven transformation to the
response such that the transformed variable is standard normal or follows
some other convenient distribution.  In this paper, we propose a framework
of multivariate conditional transformation models (MCTMs) that apply this
principle to define a novel class of multivariate distributional regression
models.  We review relevant developments in multivariate
distributional regression first before highlighting some special features of
the new method.}

The most prominent multivariate regression model is
seemingly unrelated regression (SUR), which uses a vector of
correlated normal error terms to combine several linear model regression
specifications with a common correlation structure that does not depend on
any of the covariates~\citep{Zel1962}. By construction, the model is
restricted to capture linear dependencies.  \citet{LanAdeFahSte2003} extended
SUR models by replacing the frequently used linear predictor with a
structured additive predictor, while retaining the assumption that the
linear correlation structure does not depend on the covariates.
Multivariate probit models use a latent SUR model for a multivariate set of
latent utilities that, via thresholds, are transformed to the observed
binary response vector ~\citep{Hec1978}.  The approach of
\citet{KleKneKlaLan2015} embeds bivariate SUR-type specifications into
generalised additive models for location, scale and shape \citep[GAMLSS,
][]{RigSta2005} by allowing all distribution parameters, including the
correlations, to be related to additive predictors.

Beyond SUR-type models, copulas provide a flexible approach to the
construction of multivariate distributions and regression models.  As a
major advantage, the model-building process is conveniently decomposed into
the specification of the marginals and the selection of an appropriate
copula function that defines the dependence structure \citep[see][for
reviews on copula models and their properties]{Joe1997, Nel2006}.

{There is a rich literature on conditional copula modelling.  To name just a few, \citet{VerOmeGij2011} use kernel methods to estimate the copula
parameter after having determined the marginal distributions  empirically.
Assuming that the marginal distributions are known, a  copula can be fitted
based on a local likelihood using the approach of~\citet{AcaCraYao2011}.
Bayesian inference in bivariate conditional copula models with homoscedastic
Gaussian marginals has been proposed in \citet{SabWeiCra2014} and \citet{LevCra2018}.}

Analogous to GAMLSS, bivariate copula models with parametric marginal
distributions, one-parameter copulas, as well as joint semiparametric specifications for
the predictors of all parameters of both marginal and copula models
have been developed by~\cite{GM-csda} in a penalised likelihood
framework and by \cite{KleKne2016} using a Bayesian approach.  Following
these lines,~\cite{MarRad2019} recently extended the framework to copula
link-based survival models, while \citet{Sun2019} develop a copula-based
semiparametric regression method for bivariate data under general interval
censoring.  Alternatives to simultaneous estimation are two-step procedures
that first estimate the marginals and then the copula given the marginals
and have been proposed by \eg~\cite{Vatter} and
\citet{VGAMbook} for parametric marginal
distributions and bivariate one-parameter and  copulas.  However, these
approaches are mostly limited to the bivariate case. \cite{VatNag2018} recently proposed a sequential method for conditional pair-copula constructions.

Nonparametric attempts to simultaneously study multivariate response
variables have been reported in the context of multivariate quantiles.
Because no natural ordering exists beyond univariate settings, definitions
of multivariate quantiles are challenging and there has been considerable
debate regarding their desirable properties \citep[see][for an introduction
to the different definitions]{Ser2002}.  For example, one group of
approaches draws on the concept of data depths \citep[see for
example][]{Mos2013}, utilising options for multivariate depth functions
based on distances such as Mahalanobis and Oja depths, weighted distances
or on half-spaces. However, potential quantile crossings need
further investigations to ensure a coherent model for the joint
distribution, because single quantiles only relate to local properties of a
response to covariates. For more information on depth functions and multivariate
quantiles, we refer the reader to~\citet{Cheetal2017} and
\cite{Caretal2016,Caretal2017}.

MCTMs constitute a novel and coherent approach to multivariate regression
analysis \NK{which is in many aspects different to existing approaches in
copula or nonparametric regression.} {In particular, this framework makes
six important contributions, none of which are available simultaneously in
any existing method to regression for multivariate responses:
\begin{enumerate}[i.]
\item MCTMs allow for direct estimation and inference of the entire
multivariate conditional cumulative distribution function (CDF) $F_\mY(\yvec
\mid \rx) = \dsP(\mY\leq\yvec \mid \xvec)$ of a $J$-dimensional response
vector $\mY$ given covariate information $\xvec$ under rather weak
assumptions.  A key feature of MCTMs is that they extend likelihood-based
inference in univariate conditional transformation models
\citep[CTMs,][]{moehotbue2017} to the multivariate situation in a natural
way.

\item {MCTMs can capture nonlinear aspects of covariates on all aspects of
the distribution, \eg marginal moments, marginal and joint quantiles,
dependence structures etc.} As in the case of copulas, a feature of the
model specification process is that joint distributions are constructed by
their decomposition into marginals and the dependence structure. Most existing approaches assume a constant dependence structure not varying
over the covariate space.

\item Model estimation can be performed simultaneously for all model
components, thus avoiding the need for two-step estimators that are commonly
applied in most copula-based approaches.

\item Theoretical results on optimality properties, such as consistency and
asymptotic normality are available, building on the achievements in
univariate CTMs.

\item Unlike multivariate GAMLSS, MCTMs neither require strong parametric
assumptions nor separate the model estimation process into local properties,
as in multivariate quantile regression.

\item The method scales well to situations beyond the bivariate case $J = 2$ and
readily allows for the determination of both the marginal distributions of
subsets of the response vector and the conditional distributions of some
response elements, given the others.  MCTMs are not equivalent to
copulas, however, Gaussian copulas~\citep{Son2000} with
arbitrary marginal distributions are treated as a special case in this
paper. Both the marginal distributions and the correlation parameters
of the copula can depend on covariates when such a copula model is specified
by means of an MCTM.

\end{enumerate}}

The paper is structured as follows: Section~\ref{sec:utm} provides details
on the specification of multivariate transformation models for the
unconditional case of absolutely continuous responses.  Likelihood-based
inference and optimality properties are derived in Section 3, along with an illustration on multivariate density estimation with highly non-Gaussian marginal distributions.
Section~\ref{sec:mctm} considers how multivariate conditional transformation
models may depend on covariates, and the approach is illustrated by a trivariate
analysis of childhood undernutrition indicators.  Section~\ref{sec:sim}
presents simulation-based empirical evidence on the performance of MCTMs, including examples with up to 10 response dimensions.  Finally,
Section~\ref{sec:conc} proposes directions for future research.

\section{Multivariate Transformation Models} \label{sec:utm}

\subsection{Basic Model Setup}

First, \emph{unconditional} transformation models are developed for
the joint multivariate distribution of a $J$-dimensional, absolutely
continuous random vector $\mY=(\rY_1,\ldots,\rY_J)^\top \in \dsR^J$
with density $f_\mY(\yvec)$ and CDF
$F_\mY(\yvec) = \dsP(\mY\leq\yvec)$.  These unconditional
models are then extended to the regression case in Section~\ref{sec:mctm}.

The key component of multivariate transformation models is an unknown,
bijective, strictly monotonically increasing transformation function $\h:
\dsR^J\rightarrow \dsR^J$.  This function maps the vector $\mY$, whose
distribution is unknown and shall be estimated from data, to a set of $J$
independent and identically distributed, absolutely continuous random
variables $\rZ_j \sim \Prob_\rZ, j = 1, \dots, J$ with an a priori defined
distribution $\Prob_\rZ$, such that
\[
 \h(\mY) = (\h_1(\mY), \ldots, \h_J(\mY))^\top \overset{d}{=} (\rZ_1, \ldots, \rZ_J)^\top = \mZ \in \dsR^J.
\]
For an  absolutely continuous distribution $\Prob_\rZ$ with log-concave
density $\dZ$, it can easily be shown that a unique, monotonically increasing
transformation function $\h$ exists for arbitrary, absolutely continuous
distributions of $\mY$ \citep{moehotbue2017}.  Thus, the model class is
effectively limited only by the flexibility of the specific choice of $\h$ in
an actual model.  As a default for $\Prob_\rZ$, we consider the standard
normal distribution, \ie $\rZ_j \sim \ND(0,1)$ with $\Prob_\rZ = \ND(0,1)$
and thus $\dZ = \phi_{0,1}$ and $\pZ = \Phi_{0,1}$. Enforcing independent standard normality of the
transformed response $\h(\mY)$ is computationally attractive and allows the
dependence structure to be described by a Gaussian copula (see
Section~\ref{subsec:gauss}). Alternative choices
for $\Prob_\rZ$ are discussed in Section~\ref{subsec:alt}.


Under this transformation model, the task of estimating the distribution of
$\mY$ simplifies to the task of estimating $\h$.  Because $\h$  is strictly monotonically increasing
in each element, it has a positive definite Jacobian, \ie
\begin{equation}\label{eq:jacobian}
 \left|\frac{\partial \h(\yvec)}{\partial\yvec}\right|>0.
\end{equation}
The density of $\mY$ implied by the transformation model is then
\[
 f_\mY(\yvec) = \left[\prod_{j=1}^J \dZ(h_j(\yvec))\right] \cdot \left|\frac{\partial h(\yvec)}{\partial\yvec}\right|.
\]
This form of an unconditional multivariate density with $\Prob_\rZ = \ND(0, 1)$ is called a normalizing
flow in machine learning~\citep{papamakarios2019normalizing}.
However, in this generality, the model is cumbersome in terms of both interpretation and tractability.  Thus, in the following, we introduce
simplified parameterisations of $\h$ that lead to interpretable models.

\subsection{Models with Recursive Structure}\label{MwRS}

In a first step,  we impose a triangular structure on the transformation
function $\h$ by assuming
\[
 h_j(\yvec) = h_j(y_1,\ldots,y_J) = h_j(y_1,\ldots,y_j)
\]
\ie the $j$th component of the transformation function depends only on the
first $j$ elements of its argument $\yvec$.
Consequently, this model formulation depends inherently on
the ordering of the elements in $\mY$.  Because any multivariate
distribution can be factored into a sequence of conditional
distributions, the triangular structure does not pose a limitation in
the representation of $J$-dimensional continuous distributions, as long as the
transformation functions are appropriately chosen.  Furthermore, the
triangular structure of $\h$ considerably simplifies the determinant of the
Jacobian~\eqref{eq:jacobian}, which reduces to
\[
 \left|\frac{\partial h(\yvec)}{\partial\yvec}\right| =
     \prod_{j=1}^J\left|\frac{\partial h_j(\ry_1,\ldots,y_j)}{\partial \ry_j}\right|.
\]
%
In a second step, we assume that the triangulary structured transformation
functions are linear combinations of marginal transformation functions
$\tilde{\h}_j: \dsR\rightarrow\dsR$, \ie
\[
 \h_j(\ry_1,\ldots, \ry_j) = \lambda_{j1}\tilde{\h}_1(\ry_1) + \ldots + \lambda_{jj}\tilde{\h}_j(\ry_j)
\]
where each $\tilde{\h}_j$ increases strictly monotonically and
$\lambda_{jj} > 0$ for all $j=1,\ldots,J$ to ensure the bijectivity of $\h$. Because the
last coefficient, $\lambda_{jj}$, cannot be separated from the marginal
transformation function $\tilde{\h}_j(\ry_j)$, we use the
restriction $\lambda_{jj} \equiv 1$. Thus, our parameterisation of the
transformation function $\h$ finally reads
 \begin{equation}\label{eq:finalmodel}
 \h_j(\ry_1,\ldots, \ry_j) = \lambda_{j1}\tilde{\h}_1(\ry_1) + \ldots + \lambda_{j,j-1}\tilde{\h}_{j-1}(\ry_{j-1}) + \tilde{\h}_j(\ry_j).
 \end{equation}
Each of the marginal transformation functions $\tilde{\h}_j(\ry_j)$
includes an intercept, such that no additional intercept term can be
inserted in~\eqref{eq:finalmodel}.  The Jacobian of $\h$ now further simplifies
to
 \[
 \left|\frac{\partial \h(\yvec)}{\partial\yvec}\right| =
   \prod_{j=1}^J \frac{\partial \tilde{\h}_j(\ry_j)}{\partial \ry_j} > 0,
 \]
and the model-based density function for $\mY$ is therefore
\[
 f_\mY(\yvec) = \prod_{j=1}^J  \dZ\left(\lambda_{j1}\tilde{\h}_1(\ry_1) + \ldots +
 \lambda_{j,j-1}\tilde{\h}_{j-1}(\ry_{j-1}) + \tilde{\h}_j(\ry_j)\right)  \frac{\partial
 \tilde{\h}_j(\ry_j)}{\partial \ry_j}.
\]
Summarising the model specification, our multivariate transformation model
is characterised by a set of marginal transformations $\tilde{\h}_j(\ry_j)$,
$j=1,\ldots,J$, each applying to only a single component of the vector
$\mY$, and by a lower triangular $(J\times J)$ matrix of transformation
coefficients
 \[
 \mLambda = \begin{pmatrix}
 1 & & & & 0\\
 \lambda_{21} & 1 \\
 \lambda_{31} & \lambda_{32} & 1 \\
 \vdots & \vdots & & \ddots & \\
 \lambda_{J1} & \lambda_{J2} & \ldots & \lambda_{J,J-1} & 1
 \end{pmatrix}.
 \]
Under the standard normal reference distribution $\Prob_\rZ = \ND(0, 1)$,
the coefficients in $\mLambda$ characterise the dependence structure via a
Gaussian copula, while the marginal transformation functions $\tilde{\h}_j$
allow the generation of arbitrary marginal distributions for the components
of $\mY$.  {Furthermore, the entries of $\mLambda$ have the interpretation
of entries in the inverse precision matrix of the correlation matrix between
the marginally transformed components of $\mY$, as we derive in the following.}

\subsection{Relation to Gaussian Copula Models} \label{subsec:gauss}

The relationship between multivariate transformation models and Gaussian copulas
can be made more precise by defining random variables $\tilde{\rZ}_j =
\tilde{\h}_j(Y_j)$.  Under a standard normal reference distribution
$\Prob_\rZ = \ND(0, 1)$, the vector
$\mtildeZ=(\tilde{\rZ}_1,\ldots,\tilde{\rZ}_J)^\top$ follows a zero mean
multivariate normal distribution $\mtildeZ\sim\ND_J(\nullvec_J, \mSigma)$
with covariance matrix $\mSigma = \mLambda^{-1} \mLambda^{-\top}$.  As a consequence, the
elements of $\mtildeZ$ are marginally normally distributed as
$\tilde{Z}_j\sim\ND(0,\sigma_j^2)$, where the variances $\sigma_j^2$ can be
determined from the diagonal elements of $\mSigma$.

For the transformation functions $\tilde{\h}_j$, the explicit representation
\begin{equation}\label{eq:gausscop_trafo}
 \tilde{\h}_j(\rY_j) = \Phi_{0, \sigma_j^2}^{-1}(F_j(\rY_j)) = \tilde{\rZ}_j
\end{equation}
is obtained, where $F_j(\cdot)$ is the univariate marginal CDF of $\rY_j$. In summary,
 \begin{eqnarray*}
 \dsP(\mY\le \yvec) &=& \dsP(\mtildeZ\le\ztildevec)
 = \Phi_{\nullvec, \mSigma}(\ztildevec)
 = \Phi_{\nullvec,
 \mSigma}\left[\Phi_{0,\sigma_1^2}^{-1}\left\{F_1(\ry_1)\right\}, \ldots,
 \Phi_{0,\sigma_J^2}^{-1}\left\{F_J(y_J)\right\}\right]\\
 &=& \Phi_{\nullvec, \mSigma}\left(\tilde{\h}_1(\ry_1), \dots,
 \tilde{\h}_J(\ry_J)\right).
 \end{eqnarray*}
and therefore the CDF of $\mY$ has exactly the
same structure as a Gaussian copula, except that our
representation relies on a different parameterisation of $\mSigma$ through
$\mSigma = \mLambda^{-1} \mLambda^{-\top}$ rather than a covariance matrix with unit diagonal.
This is compensated for by the inclusion of univariate Gaussians with variances
different from those acting on the marginals.  However, because
\begin{equation}\label{eq:eqi}
 \tilde{\rZ}_j/\sigma_j\sim\ND(0,1) \text{ and }
 \Phi_{0,\sigma_j^2}(\tilde{z}_j) = \Phi_{0,1}(\tilde{z}_j / \sigma_j)
\end{equation}
unconditional MCTMs with $\Prob_\rZ = \ND(0, 1)$ are equivalent to a Gaussian copula with flexible marginal distributions.
This is no longer the case of the conditional case of regression considered in
Section~\ref{sec:mctm}, where our approach can capture nonlinear aspects of
covariates on all aspects of the distribution, \eg~marginal moments,
marginal and joint quantiles, dependence structures etc.

\subsection{Some Properties of the Dependence Structure}
We highlight some properties of our MCTM that can be of interest in applied studies.
\begin{itemize}
\item The transformed vector $\tilde\mZ$ is jointly multivariate normally
distributed such that pairwise dependencies are restricted to linear
dependence through linear correlations.  Importantly however,  $\mY$ is allowed to have
a nonlinear dependence structure due to the inverse marginal
transformation functions $Y_j=\tilde h_j^{-1}(\tilde Z_j)$.  We illustrate
this feature in Figure~\ref{fig:illustration} which shows a bivariate
scatterplot with one marginally normally and one marginally gamma
distributed response component for the vector $(Y_1,Y_2)^{\top}$ (on the
left) together with the bivariate normally distributed variables $\tilde\mZ$
(on the right).

\begin{figure}[t]
\begin{center}\includegraphics[width=0.9\textwidth,angle=0]{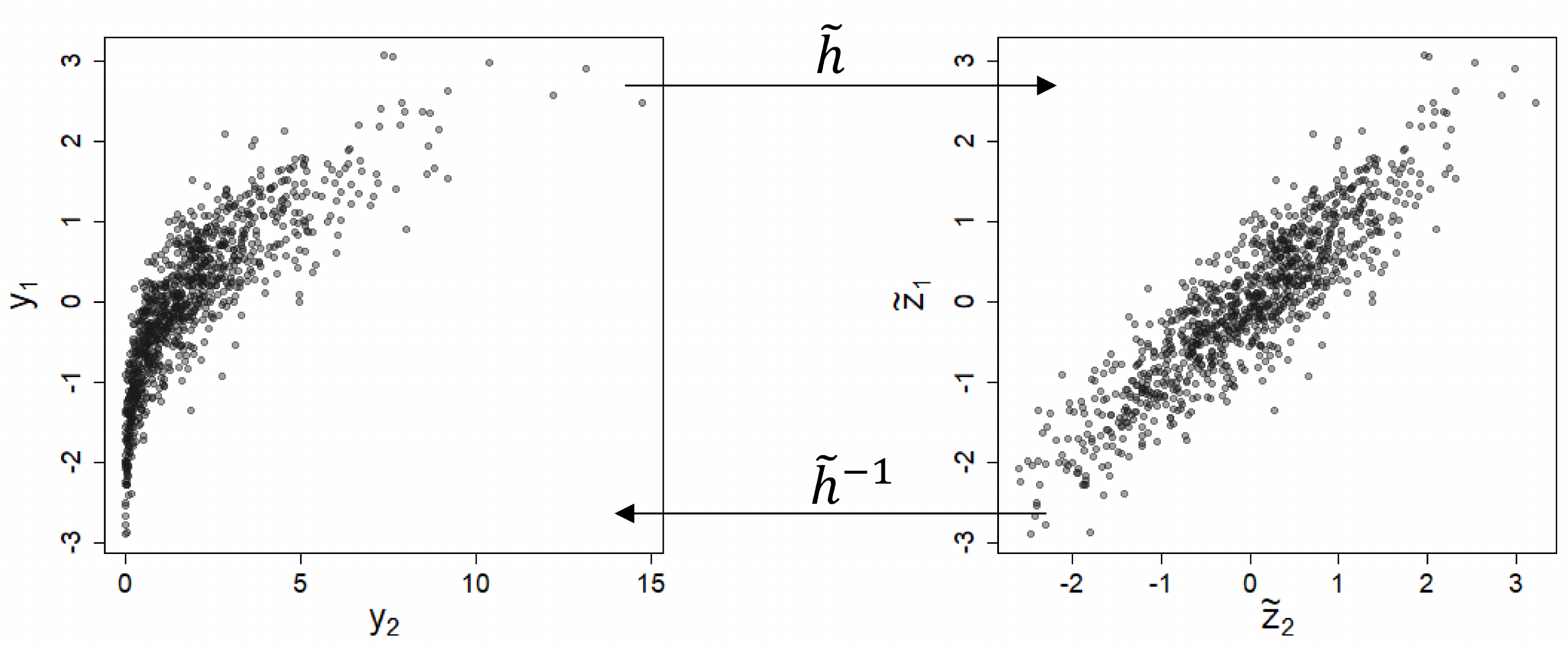}\end{center}
\caption{\footnotesize Bivariate illustration of the relation between the distributions of $\tilde\mZ$ (right) and $\mY$ (left). The distribution of $\mY$ is constructed from a Gaussian copula with correlation parameter $\rho=0.9$, standard normally distributed $Y_1$ and gamma distributed $Y_2$.}
\label{fig:illustration}
\end{figure}

\item Assuming  $\dsP_Z=\ND(0,1)$ as reference distribution, the entries in $\mLambda$ determine the conditional
independence structure between the transformed responses $\tilde{Z}_j$ and
therefore, implicitly, also the observed responses $Y_j$ as it is for a
multivariate Gaussian distribution, see the Appendix~\ref{app:theos:dep} for details.

\item Rather than looking at linear correlations, common measures of dependence in
the context of multivariate modelling are Spearman's rho $\rho^S$, Kendall's
tau $\tau^K$ and lower/upper quantile dependence $\lambda^L/\lambda^U$. These can computed in closed form using the results known for a Gaussian
copula, see again the Appendix~\ref{app:theos:dep} for formulas.  One appealing property
of these measures is that they are invariant with respect to monotonic
transformations of the marginals and we will use the $\rho^S$ later in our trivariate
application on childhood undernutrition.
\end{itemize}

\subsection{Model-Implied Marginal and Conditional Distributions}

The relationship of unconditional transformation models and Gaussian copulas can
now be employed to facilitate the derivation of model-implied marginal and
conditional distributions.  The univariate marginal distributions of elements
$Y_j$ are given by
\[
 F_j(\ry_j) = \dsP(\rY_j \le \ry_j) = \Phi_{\nullvec, \mSigma}\left(\infty, \dots, \infty,
\tilde{\h}_j(\ry_j), \infty, \dots, \infty\right) = \Phi_{0,\sigma_j^2}\left(\tilde{\h}_j(\ry_j)\right)
\]
but more general versions (\ie marginals of a subvector of $\mY$) and conditional distributions are also easily obtained, see the Appendix~\ref{app:marg:cond}.

Finally, using the marginal CDFs and densities, the marginal quantiles or moments can be derived. The latter can be computed by solving simple univariate numerical integrals, for example:
\[
 \dsE(Y_{j}) = \int F_{j}^{-1}(\Phi_{0,\sigma_{{j}}^2}(\tilde z)) \phi_{0,\sigma_{{j}}^2}(\tilde z)\mathrm{d}\tilde z,
\]
for the marginal mean.

\subsection{Alternative Reference Distributions} \label{subsec:alt}

Although we have discussed our model specification in the context of a normal
reference distribution and a Gaussian copula,
these choices can be readily modified.  In particular, if a reference distribution $\Prob_\rZ
\neq \ND(0, 1)$ is chosen, the transformation function has to be modified to
\[
 \h_j(\ry_1,\ldots, \ry_j) =
\sum_{\jmath = 1}^{j - 1} \lambda_{j\jmath}
\Phi^{-1}_{0,\sigma_\jmath^2}\left[\pZ\left\{\tilde{\h}_\jmath(\ry_\jmath)\right\}\right] +
\Phi^{-1}_{0,\sigma_j^2}\left[\pZ\left\{\tilde{\h}_j(\ry_j)\right\}\right].
\]
We therefore obtain $J$ independent random variables
$\Phi^{-1}_{0,\sigma_j^2}\left\{\pZ\left[\tilde{\h}_j(\rY_j)\right]\right\}
\sim \ND(0, \sigma^2_j)$. The model then implies marginal distributions
\[
 F_j(\ry_j) = \Phi_{\nullvec, \mSigma}\left(\infty, \dots, \infty,
\Phi^{-1}_{0,\sigma_j^2}\left\{\pZ\left[\tilde{\h}_j(\ry_j)\right]\right\},
\infty, \dots, \infty\right) = \pZ\left(\tilde{\h}_j(\ry_j)\right).
\]
Attractive alternative choices for the reference distribution are $\pZ^{-1}=\text{logit}$ and $\pZ^{-1}
=\text{cloglog}$, because regression coefficients can be interpreted as
log-odds ratios and log-hazard ratios (in fact, in the latter case, the marginal model is then a
Cox proportional hazards model), respectively.

Extensions beyond the Gaussian copula structure are also conceivable when
the linear combination of marginal transformations is replaced by nonlinear
specifications.  However, those types of models easily lead to identification problems
and do not provide direct links to existing parametric copula classes. Accordingly, we leave this topic for future research.

\section{Transformation Analysis}\label{sec:inf}

This section defines the maximum likelihood estimator and establishes its
consistency and asymptotic normality based on suitable parameterisations of
the marginal transformation functions $\tilde{\h}_j$. It closes with an illustration on bivariate density estimation for highly non-Gaussian data.

\subsection{Parameterisation of the Transformation Functions}\label{subsec:trafo}

Following~\cite{moehotbue2017}, the marginal transformation
functions $\tilde{\h}_j(\ry_j)$ are parameterised as linear combinations of the
basis-transformed argument $\ry_j$, such that $\tilde{\h}_j(\ry_j) =
\basisy_j(\ry_j)^\top \parm_j$ is monotonically increasing.  The $\dimparm_j$-dimensional basis functions $\basisy_j: \dsR
\rightarrow \dsR^{\dimparm_j}$ with basis coefficients $\parm_j$ and corresponding derivative
$\tilde{\h}^\prime_j(\ry_j) = \basisy^\prime_j(\ry_j)^\top \parm_j > 0$ are
problem-specific, see \cite{moehotbue2017} for suitable choices in different
applications. Because marginal transformation functions $\tilde{\h}_j(\ry_j)$ and therefore also the
plug-in estimators $\hat F_j$ of the marginal CDF
should be smooth with respect to $y_j$, in principle any polynomial or
spline-based basis is a suitable choice for $\basisy_j$.  The empirical results of
Sections~\ref{sec:mctm} and~\ref{sec:sim} rely on Bernstein polynomials of order
$M$; suitable choices of this parameter are discussed in Section~\ref{sim1}.
The basis functions $\basisy_j(\ry_j)$ are then densities of beta
distributions, a choice that is computationally appealing because strict
monotonicity can be formulated as a set of linear constraints on the
components of the parameters $\parm_{j}$, see~\citet{CurGho2011,Far2012} for
details. Furthermore, Bernstein polynomials of sufficiently large order $M$ can uniformly
approximate any function over an interval as a result of the Weierstrass
approximation theorem. \cite{moehotbue2017} investigate the choice of $M$ for univariate CTMs.

\subsection{Inference}\label{subsec:Inference}

In the following, we denote the set of parameters describing all marginal
transformation functions $\tilde{\h}_j, j = 1, \dots, J$ as
$\parm=(\parm_1^\top,\ldots,\parm_J^\top)^\top\in\dsR^{\sum_{j=1}^J P_j}$,
while $\lambdavec$ contains all unknown elements of $\mLambda$, such that
$\thetavec=(\parm^\top,\lambdavec^\top)^\top$ comprises all unknown model
parameters.  The parameter space is denoted as $\Theta=\lbrace\thetavec|
h\in\mathcal{H}\rbrace$, where
\begin{equation*}
\begin{aligned}
\mathcal{H}=\bigg\lbrace h:\dsR^J&\to\dsR^J \mid \h \, \text{as in (\ref{eq:finalmodel}), }\h\mbox{ strictly monotonically increasing} \bigg\rbrace
\end{aligned}
\end{equation*}
is the space of all strictly monotonic triangular transformation functions.
Consequently, the problem of estimating the unknown transformation function
$\h$, and thus the unknown distribution function $F_\mY$, reduces to the
problem of estimating the parameter vector $\thetavec$. With the
construction of multivariate transformation models, this is conveniently achieved using
likelihood-based inference.
For $\Prob_\rZ = \ND(0, 1)$, the log-likelihood contribution of a given datum
$\yvec_i = (\ry_{i1}, \dots, \ry_{iJ})^\top \in \dsR^J$, $i=1,\ldots,n$ is
\[
\ell_i(\thetavec) =
-\frac{1}{2}\sum_{j=1}^J \left(\sum_{\jmath = 1}^{j - 1}
  \lambda_{j\jmath} \basisy_\jmath(\ry_{i\jmath})^\top \parm_\jmath +
  \basisy_j(\ry_{ij})^\top \parm_j \right)^2
+ \log\left(\basisy^\prime_j(\ry_{ij})^\top \parm_j\right)
\]
with corresponding score contributions
\begin{eqnarray}\label{s1}
\frac{\partial \ell_i(\thetavec)}{\partial \parm_{k}} & = &
\sum_{j = k}^J - \left(\sum_{\jmath = 1}^{j - 1}
  \lambda_{j\jmath} \basisy_\jmath(\ry_{i\jmath})^\top \parm_\jmath +
  \basisy_j(\ry_{ij})^\top \parm_j \right) \lambda_{jk}
\basisy_k(\ry_{ik})
  + \frac{\basisy^\prime_k(\ry_{ik})}{\basisy^\prime_k(\ry_{ik})^\top \parm_{k}}
  \\\label{s2}
\frac{\partial \ell_i(\thetavec)}{\partial \lambda_{\tilde k k}} & = &
- \left(\sum_{\jmath = 1}^{\tilde k - 1}
  \lambda_{\tilde k\jmath} \basisy_\jmath(\ry_{i\jmath})^\top \parm_\jmath +
  \basisy_{\tilde k}(\ry_{i\tilde k})^\top \parm_{\tilde k} \right) \basisy_k(\ry_{ik})^\top \parm_k
\end{eqnarray}
for $k = 1, \dots, J$, $1 \le k < \tilde k\le J$ (and zero otherwise) and
with $\lambda_{jj} \equiv 1$.  We furthermore define
$\mathcal{F}_i(\thetavec)=-\frac{\partial^2
\ell_i(\thetavec)}{\partial\thetavec\partial\thetavec^\top}$ as the $i$th
contribution to the observed Fisher information.  Explicit expressions for
the entries are given in Appendix~\ref{app:fisher}.  Despite the
estimation of a fairly complex multivariate distribution with
a Gaussian copula dependence structure and arbitrary marginals, the
log-likelihood contributions have a very simple form.  In addition,
the log-concavity of $f_Z$ ensures the concavity of the log-likelihood and thus
the existence and uniqueness of the estimated transformation
function $\hat \h$.
\begin{defin} (Maximum likelihood estimator.) The maximum likelihood estimator (MLE) for the parameters of a multivariate transformation model  is given by
\begin{equation}\label{eq:MLE}
\begin{aligned}
\hat\thetavec_n &=\argmax_{\thetavec\in\mTheta}\sum_{i=1}^n \ell_i(\thetavec).
\end{aligned}
\end{equation}
\end{defin}
Based on the maximum likelihood estimator $\hat\thetavec_n$,
maximum likelihood estimators for the marginal and joint CDFs are also obtained,
by plugging in $\hat\thetavec_n$.  Specifically, the estimated marginal CDFs are given by $\hat
F_j(y_j) = \Phi_{0,\hat\sigma_j^2}(\basisy_j(\ry_j)^\top \hat{\parm}_j)$, where
$\hat\sigma_j^2$ is the $j$th diagonal entry of $\hat{\mSigma}$.
The estimated joint CDF reads
\[
 \hat F_\mY(\yvec)=\Phi_{\nullvec,\hat\mSigma}\left(
    \basisy_1(\ry_1)^\top \hat{\parm}_1, \dots,
    \basisy_J(\ry_J)^\top \hat{\parm}_J\right).
\]

\subsection{Parametric Inference}\label{subsec:param:Inference}

In this section, we discuss likelihood-based inference and establish
asymptotic results for multivariate transformation models based the
theoretical results derived in \cite{moehotbue2017} for univariate
conditional transformation models.  Assume
$\mY_1,\ldots,\mY_n\overset{\mbox{\scriptsize{i.i.d.}}}{\sim}
F_{\mY,\thetavec_0}$ where $\thetavec_0$ denotes the true parameter vector,
then the following assumptions are made:
\begin{itemize}[(A1)]
\item The parameter space $\mTheta$ is compact.
\item[(A2)]  $\dsE_{\thetavec_0}\lbrack \sup_{\thetavec\in\mTheta}[\log\{f_\mY(\mY|\thetavec)\}]\rbrack<\infty$, where $$\sup_{\lbrace
\thetavec|\,|\thetavec-\thetavec_0|>\epsilon\rbrace}\dsE_{\thetavec_0}[\log\{f_\mY(\mY|\thetavec)\}]<\dsE_{\thetavec_0}[\log\{f_\mY(\mY|\thetavec_0)\}].$$
\item[(A3)]
\[
\dsE_{\thetavec_0}\left(\sup_{\thetavec}\left|\left|\frac{\partial\log(f_\mY(\mY|\thetavec))}{\partial\thetavec}\right|\right|^2\right)<\infty.
\]
\item[(A4)]$\dsE_{\thetavec_0}(\mathcal F(\thetavec))$ is nonsingular.
\item[(A5)] $0<f_Z<\infty$, $\sup|f^\prime_Z|<\infty$, $\sup|f_Z^{\prime
\prime}|<\infty$.
\end{itemize}

{\begin{remark}
Assumption (A1) is made for convenience, and relaxations of such a condition are given in~\cite{vv1998}. The assumptions in (A2) are rather
weak: the first one holds if the functions $\basisy$ are not arbitrarily ill-posed, and the second one
holds if the function $\dsE_{\thetavec_0}\lbrack \sup_{\thetavec\in\mTheta}[\log\{f_\mY(\mY|\thetavec)\}]\rbrack$ is strictly convex in $\thetavec$ (if the assumption would not
hold, we would still have convergence to the set $\dsE_{\thetavec_0}\lbrack \sup_{\thetavec\in\mTheta}[\log\{f_\mY(\mY|\thetavec)\}]\rbrack$). Assumptions (A3)--(A5) are needed to derive the asymptotic distribution.
\end{remark}}
\begin{cor}\label{cor1}
Assuming (A1)--(A2), the sequence of estimators $\hat\thetavec_n$ converges in probability $\hat\thetavec_n\overset{\dsP}{\to}\thetavec_0$ for $n\to\infty$.
\end{cor}
The proof of Corollary~\ref{cor1} follows from Theorem~5.8 of~\cite{vv1998}.
\begin{cor}\label{cor2}
Assuming (A1)--(A5), the sequence of estimators $\sqrt{n}(\hat\thetavec_n-\thetavec_0)$ is asymptotically normally distributed with covariance matrix
\begin{equation}\label{eq:as:cov}
\left(\dsE_{\thetavec_0}\left(-\frac{\partial^2\log(f_\mY(\mY|\thetavec))}{\partial\thetavec\partial\thetavec^\top}\right)\right)^{-1}.
\end{equation}
\end{cor}
\begin{proof}
By further assumption, $\sqrt{f_\mY}$ is continuously differentiable in $\thetavec$ for all $\yvec$ and
\[
\dsE_{\thetavec_0}\left(\left\lbrack\frac{\partial\log(f_\mY(\mY|\thetavec))}{\partial\thetavec}\right\rbrack\left\lbrack\frac{\partial\log(f_\mY(\mY|\thetavec))}{\partial\thetavec}\right\rbrack^\top\right)
\]
is continuous in $\thetavec$, due to \eqref{s1},\eqref{s2}. Thus, $F_{Y,\thetavec_0}$ is differentiable in quadratic mean by Lemma 7.6 of~\citet{vv1998}. Based on Assumptions (A3)--(A5) and Corollary~\ref{cor1}, the claim hence follows from Theorem 5.39 of~\citet{vv1998}.
\end{proof}

{{\begin{remark}
Similar as in the univariate case~\citep{moehotbue2017}, Corollaries~\ref{cor1},~\ref{cor2} also extend to the conditional regression models considered in Section~\ref{sec:mctm}.
\end{remark}}
}

\subsection{Parametric Bootstrap}\label{sec:pb}

The asymptotic results allow, at least in principle, the derivation of
confidence intervals also for transformed model parameters, by using the
Delta-rule.  However, many quantities of practical interest, such as the
correlation matrix of the implied Gaussian copula or the densities of marginal
distributions, are indeed highly nonlinear functions of the parameter vector
$\thetavec$.  Accordingly, a parametric bootstrap is a more
promising alternative.

Because the proposed multivariate transformation models
allow a direct evaluation of estimated joint CDFs, drawing bootstrap samples from the joint distribution is
straightforward.  For $\Prob_\rZ = \ND(0, 1)$ and with estimated marginal
transformation functions $\hat{\tilde \h}_j(\ry_j) = \basisy_j(\ry_j)^\top
\hat{\parm}_j, j = 1, \dots, J$ and estimated covariance matrix
$\hat\mSigma=\hat\mLambda^{-1}\hat\mLambda^{-\top}$, the parametric
bootstrap can be implemented by Algorithm~\ref{algo1}.
\begin{algorithm}
\caption{ Parametric bootstrap\vspace{0.15cm}}\label{algo1}
 Given $\hat{\tilde \h}=(\hat{\tilde \h}_1(\ry_1),\ldots,\hat{\tilde \h}_J(\ry_J))^\top$ and $\hat\mLambda$:\hfill\hphantom{adsfffffffffffffff}
\begin{algorithmic}[1]
\FOR{$b=1,\ldots,B$}
\STATE Generate
\[
\mZ_{1,b},\ldots,\mZ_{nb},\quad b=1,\ldots,B,\quad \mZ_{ib}\sim\ND(\nullvec,\mI_J), \quad i=1,\ldots,n.
\]
\STATE Compute
\[
\tilde\mZ_{1b},\ldots,\tilde\mZ_{nb},\quad \tilde\mZ_{ib}=\hat\mLambda^{-1}\mZ_{ib}.
\]
\STATE Compute recursively
\[
\mY_{1b},\ldots,\mY_{nb},\quad Y_{ijb}=\hat{\tilde \h}_j^{-1}(\tilde{Z}_{ijb}), \quad j=1,\ldots.J.
\]
\STATE Re-fit the model to obtain $\hat{\tilde \h}_{(b)},\hat\mSigma_{(b)}$.
\STATE Compute summaries using $\hat{\tilde \h}_{(1)},\ldots,\hat{\tilde \h}_{(B)}$, $\hat\mSigma_{(1)},\ldots,\hat\mSigma_{(B)}$
\ENDFOR
\end{algorithmic}
\end{algorithm}
\normalsize The inverse $\hat{\tilde \h}_j^{-1}$ exists because the estimated marginal
distribution function is strictly monotonically increasing. For simple
basis functions $\basisy_j$ (for example, linear functions), the inverse can
be computed analytically. For more complex basis functions, numerical
inversion has to be applied.

\subsection{Illustration: Bivariate Density Estimation} \label{sec:cars}

Unconditional MCTMs can be employed for multivariate density estimation.
For the famous 1920s cars data \citep{Ezekiel_1930} consisting of speed
and distance needed to stop for $50$ cars, the bivariate distribution was
estimated from an unconditional bivariate transformation model with
$\Prob_\rZ = \ND(0, 1)$, order $M = 6$ of Bernstein polynomials for the two
transformation functions, and a constant parameter $\lambda \in \RR$.  The
model is equivalent to a Gaussian copula, however, the marginal
distributions are highly non-Gaussian and were estimated by maximum
likelihood simultaneously with the correlation parameter $\lambda$.  The fit
of the bivariate density contours and marginal densities is given in
Figure~\ref{fig:cars}, which shows that the dependence between speed and distance is clearly nonlinear.  We obtained $\hat{\lambda} = -1.633$ (SE $0.273$),
which corresponds to a Pearson correlation of $0.853$ (thus a rank
correlations $\rho^S=0.8415$) and hence a highly
positive correlation between speed and distance after transformation to
normality.

\begin{figure}[t]
\begin{center}\includegraphics[width=0.5\textwidth,angle=0]{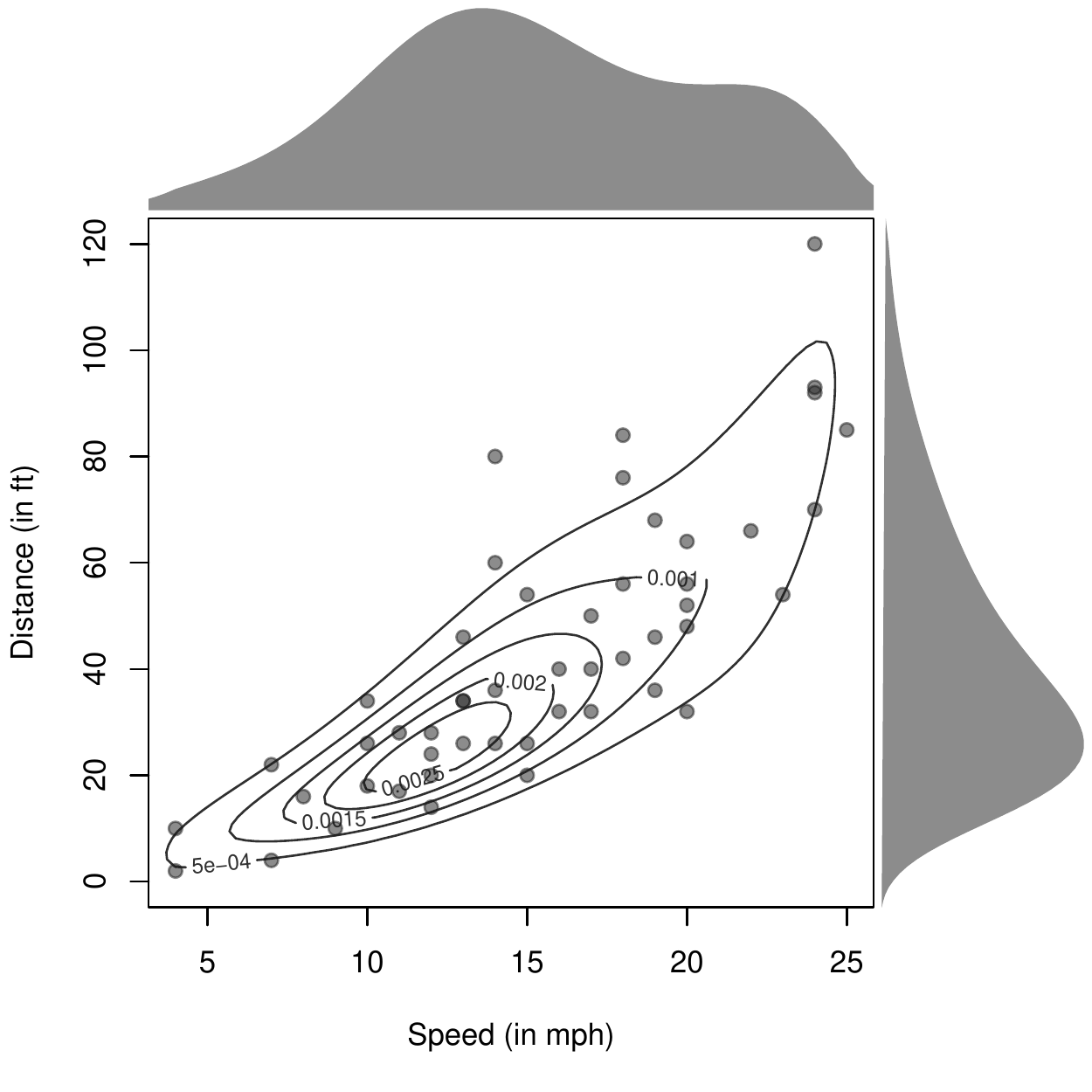}\end{center}
\caption{ Bivariate density estimation. Scatterplot of
         speed of cars and the distance needed to stop. Contours visualise
         the estimated joint density, corresponding marginal distributions
         are shaded in grey (top and right).
         \label{fig:cars}}
\end{figure}


\section{Extensions to Multivariate Regression} \label{sec:mctm}

\subsection{Multivariate Conditional Transformation Models} \label{subsec:mctm}

Multivariate regression models, \ie models for conditional multivariate
distributions given a specific configuration of covariates $\rX = \rx$, can be derived
from the unconditional multivariate transformation models introduced in
Section~\ref{sec:utm}.  The transformation function $\h$ has to be extended
to include a potential dependency on covariates $\mX$, and the
corresponding joint CDF $F_{\mY \mid
\mX=\xvec}$ is defined by a conditional transformation function
$\h(\yvec \mid \xvec)=(\h_1(\yvec \mid \xvec),\ldots,\h_J(\yvec \mid \xvec))^\top$.
By extending the unconditional transformation function
(\ref{eq:finalmodel}), we define the $J$ components of a multivariate conditional transformation function given
covariates $\rx$ as
\[ \h_j(\yvec \mid \rx) = \sum_{\jmath = 1}^{j - 1}
  \lambda_{j\jmath}(\rx) \tilde{h}_\jmath(\ry_\jmath \mid \rx) + \tilde{h}_j(\ry_j \mid \rx)
\]
where $\lambda_{j\jmath}(\rx)$ and $\tilde{\h}_j(\ry_j \mid \rx)$ are again
expressed in terms of basis function expansions.

For the marginal (with respect to the response $\ry_j$) conditional (given
covariates $\rx)$ transformation functions, this leads to a parameterisation
\[
 \tilde{h}_j(\ry_j \mid \rx) = \basisyx_j(\ry_j,\xvec)^\top \parm_j
\]
where the basis functions $\basisyx_j(\ry_j,\xvec)$, in general, depend on
both element $\ry_j$ of the response and the covariates $\xvec$. These can, for
example, be constructed as a composition of the basis functions
$\basisy_j(y_j)$ for only $\ry_j$ from the previous section combined with a
basis $\basisx_j(\xvec)$ depending exclusively on $\xvec$.  Specifically,
a purely additive model results from
$\cvec_j=(\basisy_j^\top,\basisx_j^\top)^\top$, and a flexible interaction
from the tensor product
$\cvec_j=(\basisy_j^\top\otimes\basisx_j^\top)^\top$.
Response-varying coefficients in distributional regression, or time-varying effects in survival
analysis, correspond to a basis $\cvec_j = (\basisy_j^\top \otimes (1,
\rx^\top)^\top)^\top$, \NK{see also Section~\ref{subsec:mctm2}}. A simple
linear transformation model for the marginal conditional distribution
can be parameterised as
\[
\basisyx_j(\ry_j,\xvec)^\top \parm_j = \basisy_j(\ry_j)^\top \parm_{j,1} -
\rx^\top \shiftparm_j,
\]
with parameters $\parm_j=(\parm_{j,1}^\top,\shiftparm_{j}^\top)^\top$. The
model restricts the impact of the covariates to a linear shift $\rx^\top
\shiftparm_j$.  For arbitrary
choices of $\Prob_\rZ$, the marginal distribution, given covariates $\mX =
\rx$, is then a marginal linear transformation model
\begin{eqnarray*}
\Prob(\rY_j \le \ry_j \mid \rX = \rx) & = &
  \Phi_{0, \sigma^2_j}\left(\tilde{\h}_j(\ry_j \mid \rx)\right) =
  \Phi_{0, \sigma^2_j}(\Phi^{-1}_{0, \sigma^2_j}(\pZ(\basisy_j(\ry_j)^\top \parm_{j,1} -
        \rx^\top \shiftparm_{j})) \\
& = & \pZ(\basisy_j(\ry_j)^\top \parm_{j,1} - \rx^\top \shiftparm_{j})
\end{eqnarray*}
and, consequently, the regression coefficients $\shiftparm_{j}$ can be directly
interpreted as marginal log-odds ratios ($\pZ^{-1} = \text{logit}$) or
log-hazard ratios ($\pZ^{-1} = \text{cloglog}$; this is a marginal Cox
model).  Details of the parameterisations $\basisyx_j(\ry_j,\xvec)^\top \parm_j$
for the marginal transformation functions and a discussion of the practical aspects in different
areas of application are provided in \citet{moehotbue2017}.

For practical applications, an important and attractive feature of multivariate
transformation models for multivariate regression is the possible dependency
of $\mLambda$ on covariates $\rx$.  Thus, the dependence structure of $\mY$
potentially changes as a function of $\rx$, if suggested by the data.  This
feature is implemented by covariate-dependent coefficients of
$\mLambda(\rx)$. A simple linear model of the form
\[
\lambda_{j\jmath}(\rx) = \alpha_{j\jmath} + \rx^\top \gammavec_{j\jmath},
\quad 1 \le \jmath < j \le J
\]
is one option. The case $\gammavec_{j\jmath} = \bold{0}$ implies that
the correlation between $\rY_j$ and $\rY_\jmath$ does not depend on $\rx$.
More complex forms of additive models would also be conceivable.
Of course, the number of parameters grows quadratically in $J$, such that
models that are too complex may require additional penalisation terms in the likelihood.

\subsection{Application: Trivariate Conditional Transformation Models for Undernutrition in India} \label{subsec:mctm2}

To illustrate several practical aspects of the parameterisation and
interpretation of MCTMs, we present a trivariate analysis of undernutrition in India
in the following.  Childhood undernutrition is among the most urgent
problems in developing and transition countries.  A rich database available
from Demographic and Health Surveys (DHS, \url{https://dhsprogram.com/})
provides nationally representative information about the health and
nutritional status of populations in many of those countries.  Here we use
data from India that were collected in 1998.  Overall, the data set
comprised 24,316 observations, after pre-processing of the data. For the latter, we use the same
steps as in~\cite{FahKne2011}, see the documentation available at \url{http://www.smoothingbook.org} together with further details on the pre-processing steps. We used three indicators, \textit{stunting},
\textit{wasting} and \textit{underweight}, as the trivariate response
vector, where \textit{stunting} refers to stunted growth, measured as an
insufficient height of a child with respect to age, while
\textit{wasting} and \textit{underweight} refer to insufficient weight for
height and insufficient weight for age, respectively.  Hence
\textit{stunting} is an indicator of chronic undernutrition,
\textit{wasting} reflects acute undernutrition and \textit{underweight}
reflects both.  Our aim was to model the joint distribution of
\textit{stunting}, \textit{wasting} and \textit{underweight} conditional
upon the age of the child.
{To the best of our knowledge, there is no implementation available that could estimate the dependence structure and the marginal distributions nonparametrically and conditional on covariates beyond a trivariate normal distribution (which is implemented in the \textsf{R} add-on package \citep{pkg:VGAM}).}

\paragraph{Model Specification.}
\NK{The focus of our analysis was the variation in the trivariate undernutrition
process with respect to the age of the child (in months). To be flexible in the marginal distributions and the dependence structure, we specify response-varying marginal models for $\tilde{\h}_j(\ry_j \mid \text{age})$ of the form
%
\[
\tilde{\h}_j(\ry_j \mid \text{age}) = \basisy_j(\ry_j)^\top \parm_{j,1} -
    \basisy_j(\ry_j)^\top\ \betavec_{j} \times \text{age}, \quad j
\in \{\text{stunting}, \text{wasting}, \text{underweight}\},
\]
%
%
while the coefficients of $\Lambda$ are parameterised through
\[
\lambda_{j\jmath}(\text{age}) = \basisx(\text{age})^\top
\gammavec_{j\jmath}, \quad \jmath < j \in \{\text{stunting}, \text{wasting}, \text{underweight}\}.
\]
We choose the normal reference distribution $\Prob_\rZ = \ND(0, 1)$ and the basis functions
$\basisy_j$ and $\basisx(\text{age})$ are Bernstein polynomials of order six
(Section~\ref{sim1} gives a rational for choosing this default).
Furthermore, the parameters $\parm_{j, 1}$ were estimated
under the constraint $\mD \parm_{j,1} > \nullvec$, where $\mD$ is a difference
matrix. This leads to monotonically increasing estimated marginal
transformation functions $\tilde{\h}_j$ \citep{moehotbue2017}. No such shape
constraint was applied to functions of age, \ie the parameters $\betavec_{j}$
and $\gammavec_{j\jmath}$ were estimated unconstrained.}

\paragraph{Results for Marginal Distributions.}

Figure~\ref{fig:margins} depicts the estimated marginal conditional CDFs $F_j(\ry_j \mid \text{age})$ (first row) and marginal densities $f_j(\ry_j \mid \text{age})$ (second row), with the different
colours indicating the ages of the children.  Clearly, the shapes
of the marginals differ for the three indicators, where the differences are mostly restricted to a simple shift effect for stunting, while varying amounts of asymmetry are present for stunting and even more complex changes in the shape of the distribution are identified for underweight.

\begin{figure}[t]
\centering\includegraphics[width=0.95\textwidth,angle=0]{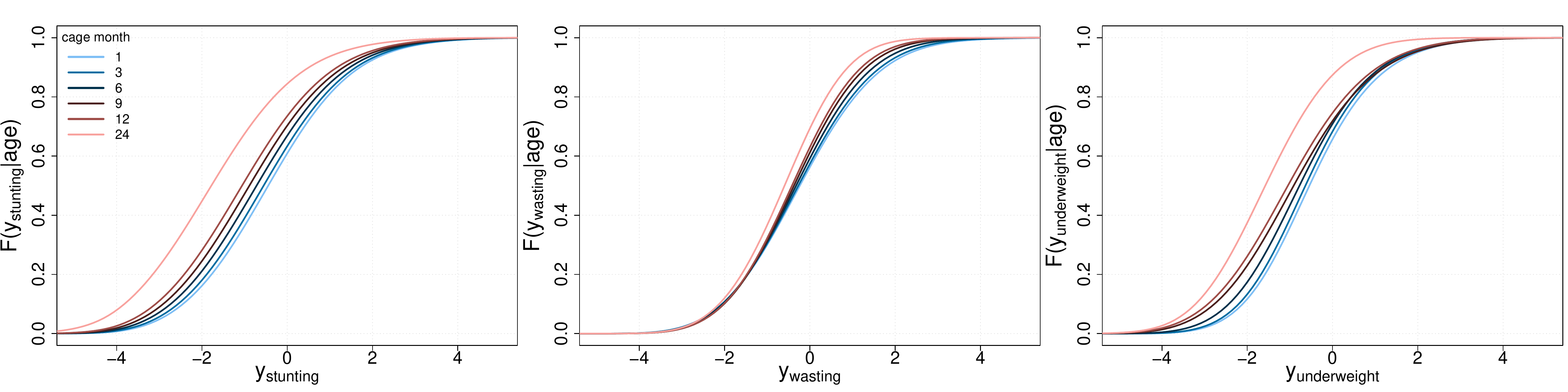}\\
\centering\includegraphics[width=0.95\textwidth,angle=0]{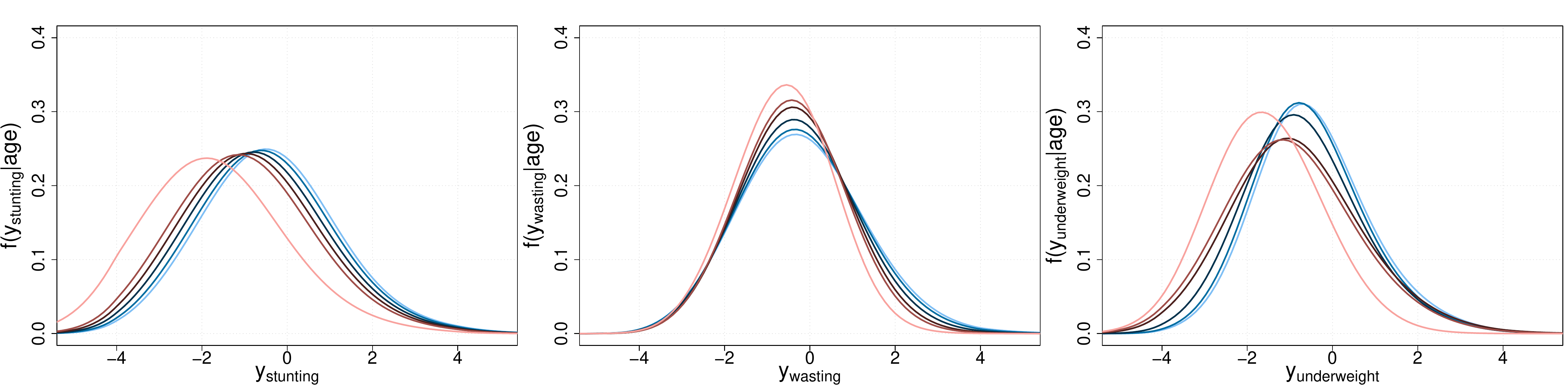}
\caption{ Undernutrition.  Estimated marginal conditional CDFs $F_j(\ry_j \mid \text{age})$, (first row) and marginal densities $f_j(\ry_j \mid \text{age})$ (second row), $j \in
\{\text{stunting}, \text{wasting}, \text{underweight}\}$ for selected ages in months.}
\label{fig:margins}
\end{figure}

\paragraph{Results for the Dependence Structure.}

Figure~\ref{fig:undernutrition} depicts the conditional rank
correlations $\rho^S$ between \text{stunting}, \text{wasting} and \text{underweight}
as functions of $\text{age}$ along with the point estimates and 95\% confidence intervals
obtained from $B=1,000$ parametrically drawn bootstrap samples (see
Algorithm~\ref{algo1} in Section~\ref{sec:pb}).
The rank correlation between $\text{stunting}$ and $\text{wasting}$ is initially
negative around $-0.4$ for young children and then approaches zero with increasing age of the child.
This finding is in line with the study of~\citet{KleKneKlaLan2015}, who
reported the results of a bivariate analysis based on normal and $t$
distributions. In our study, the remaining dependencies were positive and the variation in the rank
correlation over age was stronger for the relationship between wasting and
underweight compared to stunting and underweight, which  varied only between $0.6$ and
$0.75$ whereas the variations of the remaining rank correlations
explained by the age of the child were more substantial.

{The parametric bootstrap requires re-estimation of the model $B$ times which can be time consuming without parallelisation. A faster alternative in cases where $n$ is large is to draw samples $\hat\thetavec_{(b)}$, $b=1,\ldots,B$ from the asymptotic normal distribution derived in Corollary~\ref{cor2} with mean vector equal to the MLE from~\eqref{eq:MLE} and covariance matrix~\eqref{eq:as:cov}. These samples can then be used to compute $\hat{\tilde \h}_{(1)},\ldots,\hat{\tilde \h}_{(B)}$, $\hat\mSigma_{(1)},\ldots,\hat\mSigma_{(B)}$.
Because in our example $n$ is rather large we compared both versions of confidence intervals and found them to be very similar (compare Figure~\ref{fig:undernutrition:both} in Appendix~\ref{app:undernutrition}).}

\begin{figure}[t]
\centering\includegraphics[width=0.99\textwidth,angle=0]{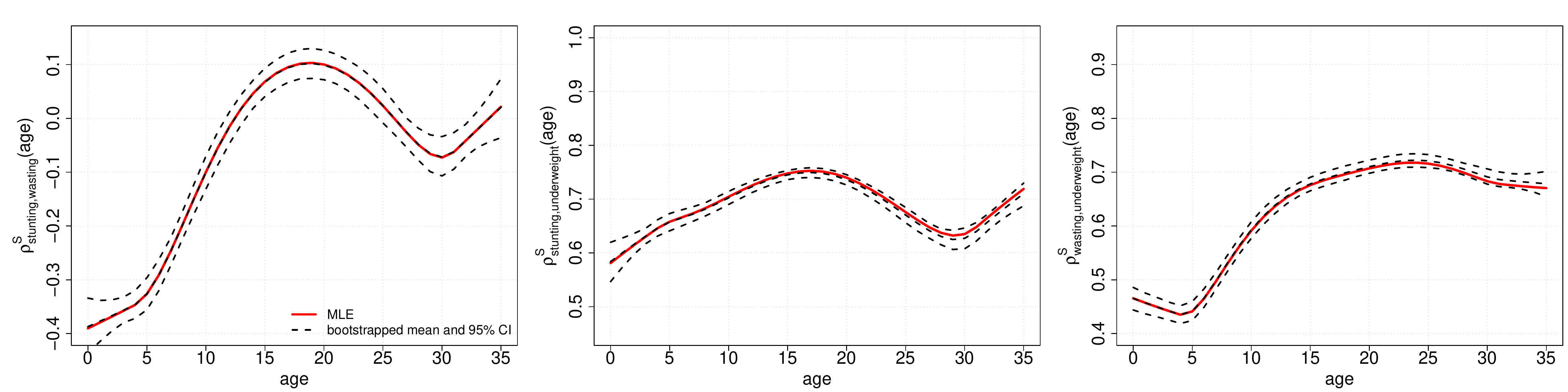}
\caption{ Undernutrition.  Spearman's rho
$\rho_{12}^S(\text{age})=\rho^S(y_{\text{\scriptsize{stunting}}},y_{\text{\scriptsize{wasting}}}|\text{age})$,
$\rho_{13}^S(\text{age})=\rho^S(y_{\text{\scriptsize{stunting}}},y_{\text{\scriptsize{underweight}}}|\text{age})$
and
$\rho_{23}^S(\text{age})=\rho^S(y_{\text{\scriptsize{wasting}}},y_{\text{\scriptsize{underweight}}}|\text{age})$.
Shown are the maximum likelihood estimates (solid red line),
the bootstrapped mean estimate and the 95\% bootstrapped
confidence intervals (dashed black lines).} \label{fig:undernutrition}
\end{figure}

\section{Empirical Evaluation}\label{sec:sim}

In this section, we provide empirical evidence on the performance of our
MCTMs.  In Section~\ref{sim1} we demonstrate that the performance of our flexible
transformation model was highly competitive relative to two parametric and
correctly specified alternatives.  Hence, our model is useful not only when
parametric assumptions about the marginals are questionable.  In
Section~\ref{sim2}, a trivariate example demonstrates that our model is also
applicable to situations beyond the bivariate case, as underpinned by five and ten-dimensional illustrations in Section~\ref{sim3}.  To the best of our
knowledge, there are currently no directly competing models that allow for a
similar flexibility.

\subsection{Bivariate Simulation}\label{sim1}
\paragraph{Simulation Design.}

We simulated $R=100$ data sets of size $n=1,000$, following a method similar to that used in the parametric bootstrap procedure:
\begin{enumerate}
 \item Covariate values $x$ were simulated as i.i.d.~variables, where $x\sim\UD[-0.9,0.9]$.
 \item The latent variables $\tilde{\zvec}_{ir}\in\dsR^2$ were generated as
 $$
\tilde{\zvec}_{ir}=\mLambda_i^{-1}\zvec_{ir},\quad i=1,\ldots,n;\,r=1,\ldots,R
 $$
 with
 $$
 \zvec_{ir}\sim\ND(\nullvec,\mI_2) \mbox{ and } \mLambda_i=\left(\begin{matrix} 1 & 0 \\ x_{ir}^2 & 1\end{matrix}\right),
 $$
 such that
 \[
 \mbox{Cov}(\tilde z_{i1},\tilde z_{i2}|x_i)\equiv\mSigma_i(x_i) = \begin{pmatrix}
 1 & -x_i^2\\
 -x_i^2 & 1+x_i^4
 \end{pmatrix}.
 \]
 \item From the latent variables, the observed responses were computed as
 \[
 \yvec_{ir}=[F_{1}^{-1}\{\Phi_{0,1}(\tilde z_{ir,1})\},F_{2}^{-1}\{\Phi_{0,\sigma_{i2}^2}(\tilde z_{ir,2})\}]^\top,
 \]
where $\sigma_{i2}^2=1+x_i^4$ and $F_1$ and $F_2$ are the CDFs of two Dagum
distributions with parameters $a_1=\exp(2),b_1=\exp(1),p_1=\exp(1.3)$ and
$a_2=\exp(1.8),b_2=\exp(0),p_2=\exp(0.9)$, respectively.  Note that the CDF
of an unconditional Dagum distribution~\citep{Kleiber.1996} reads
 \[
 F(y) = \left(1+\left(\frac{y}{b}\right)^{-a}\right)^{-p},\quad\mbox{ for } y>0 \mbox{ } a>0,b>0,p>0.
 \]

\end{enumerate}
This model specification is equivalent to a Gaussian copula model with Dagum
marginals, but by its construction, the first marginal is independent of
the covariate $x$, while the scale parameter $b_2$ of the second marginal varies as a function
of $x$.

As competitors for MCTMs, we considered Bayesian structured additive
distributional regression models \citep{KleKneLanSoh2015}, as implemented in
the software package {BayesX} \citep{BelBreKleKneLanUml2015}, and vector
generalised additive models \citep[\texttt{VGAM},][]{VGAMbook}, as
implemented in the corresponding \textsf{R} add-on package \citep{pkg:VGAM}.
For \texttt{VGAM} and {BayesX}, we employed the true specification, \ie a
Gaussian copula with correlation parameter
$\rho(x_i)=\nicefrac{-\lambda(x_i)}{\sqrt{1+\lambda(x_i)^2}}$ and Dagum
marginals, in which the parameter $b_2$ of the second marginal depends on
$x$ but the first marginal as well as the parameters $a_2$ and $p_2$ did
not.  For BayesX, both the predictor for $b_2$ and the correlation parameter $\rho$ of
the Gaussian copula were specified using cubic B-splines with $20$ inner
knots on an equidistant grid in the range of $x$ with a second-order random
walk prior \citep[following suggested default values by][]{LanBre2004}; the
other parameters of the marginals were estimated as constants.

Because \texttt{VGAM} does not allow for simultaneous estimation of the
marginals and the dependence structure, we first estimated the Dagum margins
with constant parameters $a_1,b_1,p_1,a_2,p_2$ and covariate-dependent
parameters $b_2$.  The copula predictor was then estimated with plugged-in
estimates of the margins, using cubic B-splines according to the
\texttt{sm.ps} function of the package.

For the multivariate transformation models (denoted as \texttt{MCTM-6/6}), we
employed Bernstein polynomials of order six (as in
Section~\ref{subsec:mctm2}) for both the transformation functions ($\tilde{\h}_1$
and $\tilde{\h}_2$) and the parameter $\lambda$.  Because of the monotonicity constraints on $\tilde{\h}_1$ and
$\tilde{\h}_2$, the order of the corresponding Bernstein polynomials can be
larger without decreasing model performance \citep{moehotbue2017,
Hothorn_2018_JSS}.  In contrast, too large values of a Bernstein polynomial
for $\lambda$ will result in overly erratic estimates with negative impact
on model performance. We demonstrate this effect empircially by two
additional MCTMs with order $M = 3$ for $\lambda$ and with orders $M = 6$
(\texttt{MCTM-6/3}), and $12$ (\texttt{MCTM-12/3}), for the
transformation functions $\tilde{\h}_1$ and $\tilde{\h}_2$.

\begin{figure}[t]
\centering\includegraphics[width=0.99\textwidth,angle=0]{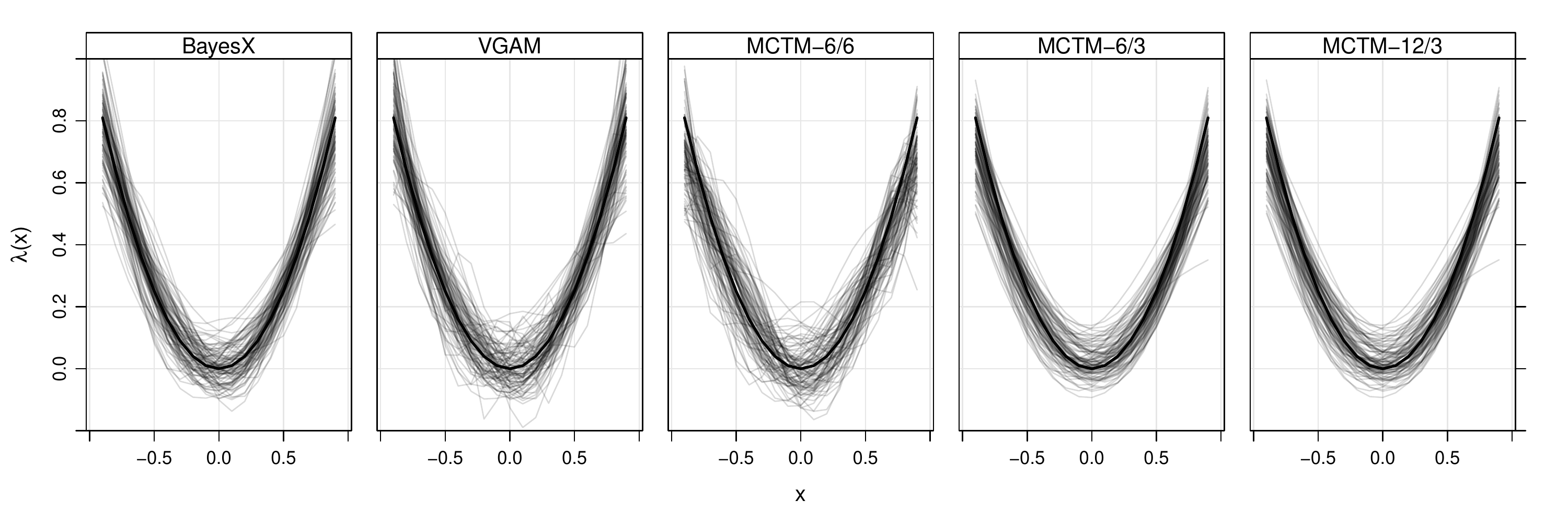}
\caption{ Bivariate simulation: Function estimates $\hat\lambda(x)$ for the effect $\lambda(x)=x^2$ on the correlation parameter. The black line is the true function and the grey lines are the estimates of the $R=100$ replicates.}
\label{fig:sim2d}
\end{figure}

\begin{figure}[t]
\centering\includegraphics[width=0.4\textwidth,angle=-90]{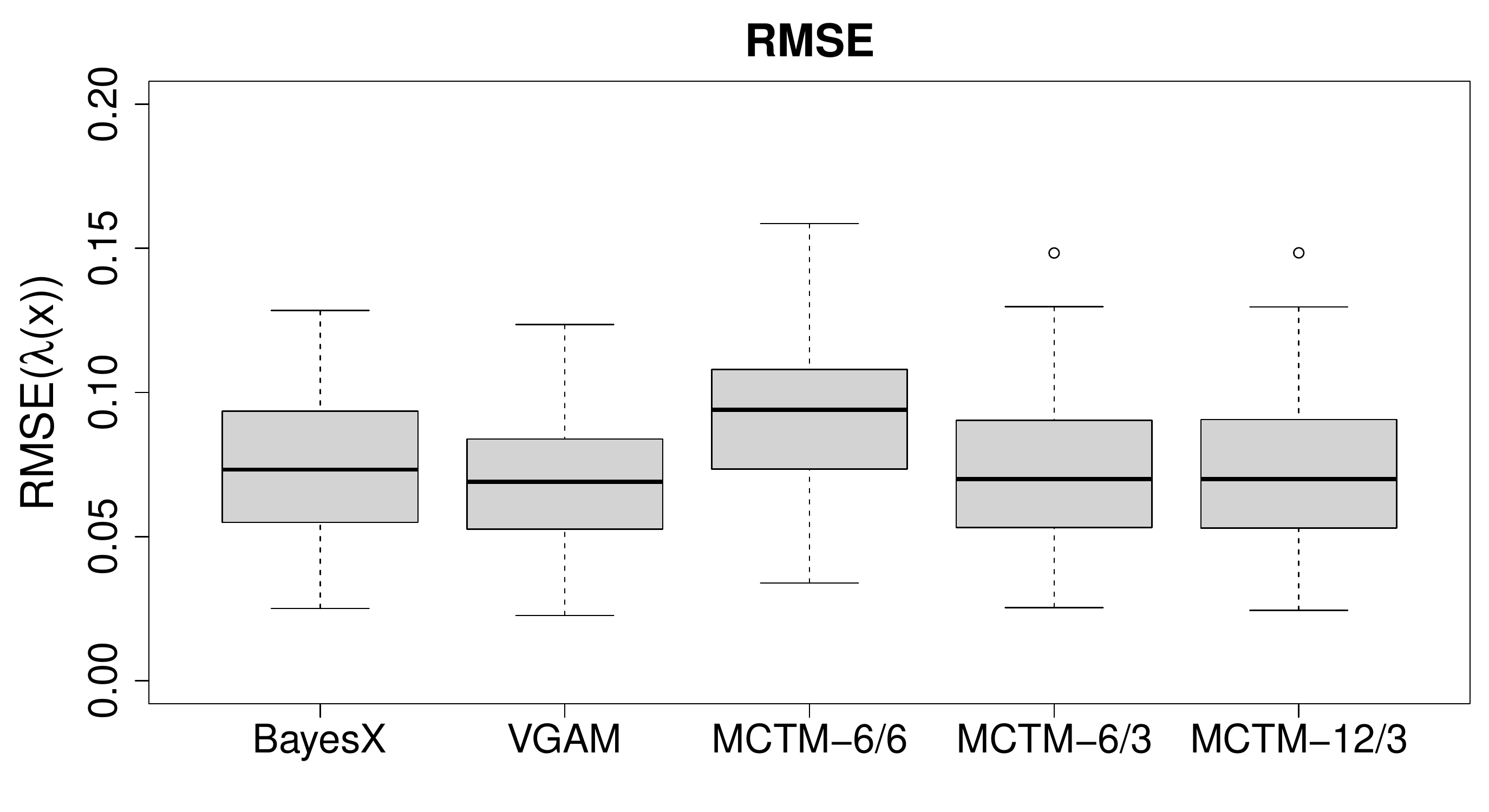}
\caption{ Bivariate simulation: RMSE($\lambda(x),\hat\lambda(x))$  for BayesX, \texttt{VGAM} and \texttt{MCTM} with Bernstein bases of order 6, 6, and 12 for the marginals and 6, 3, and 3 for $\lambda(x)$ respectively.}
\label{fig:sim2dmse}
\end{figure}

\paragraph{Measures of Performance.}
Let $\hat\lambda^{(r)}(x)$ be the estimate of the  lower triangular element of $\mLambda^{(r)}$ obtained from data replicate $r=1,\ldots,R$.
To evaluate the performance of the three competing methods, we investigated the function estimates $\hat\lambda^{(r)}(x)$ relative to $\lambda(x)$ as well as the root mean squared errors   $\mbox{RMSE}(\lambda,\hat\lambda^{(r)})=\sqrt{(\lambda(x_g)-\hat\lambda^{(r)}(x_g))^2}$ on a grid of length $G=100$ within the range of $x$.

\paragraph{Results.}
Figure~\ref{fig:sim2d} shows the estimates for
$\lambda(x)=x^2$ of the $100$ simulated data sets for BayesX (first panel), \texttt{VGAM} (second panel), and \texttt{MCTM} (last three panels).  All three models reproduced the
general functional form correctly.  However, BayesX yielded the most
reasonable smoothing properties, while \texttt{VGAM} has the wiggliest curves. 
Although BayesX and \texttt{VGAM} employed the correct model specification in terms of the parametric distribution assumption for the
marginal distributions and the correlation parameter, the performance of
\texttt{MCTM} is competitive in terms of the RMSE
(Figure~\ref{fig:sim2dmse}) without the requirement to either estimate the marginal distributions in a first step and plug the empirical copula data in to obtain the dependence structure (as for VGAM) or to specify predictors for parametric marginal distributions (as for BayesX). Both requirements are restrictive in practice
because typically it is impossible to pick the \lq correct\rq\,parametric distribution that exactly matches the marginal distributions of the underlying random variables.

A larger value of $M = 12$ for the transformation functions did not lead to
degraded performance, however, a less flexible parameterisation of $\lambda$
was better able to recover the quadratic function.

\subsection{Trivariate Simulation}\label{sim2}

\paragraph{Simulation Design.}
We employed a similar setting as in the previous section, with $R=100$ data sets of size $n=1,000$,  $x\overset{\mbox{\scriptsize{i.i.d.}}}{\sim}\UD[-0.9,0.9]$, but Steps~2 and~3 of the simulation design in Section~\ref{sim1} were extended to three dimensions.
The latent variables $\tilde{\zvec}_{ir}\in\dsR^3$ were generated as
 $$
\tilde{\zvec}_{ir}=\mLambda_i^{-1}\zvec_{ir},\quad i=1,\ldots,n;\,r=1,\ldots,R
 $$
 with
 $$
 \zvec_{ir}\sim\ND(\nullvec,\mI_3) \mbox{ and } \mLambda_i=\left(\begin{matrix} 1 & 0 & 0\\ x_{ir}^2 & 1 & 0\\ -x_{ir} & x_{ir}^3-x_{ir} & 1\\ \end{matrix}\right).
 $$
Consequently,
 \[
 \mbox{Cov}(\tilde z_{i1},\tilde z_{i2},\tilde z_{i3}|x_i)\equiv\mSigma_i(x_i) = \begin{pmatrix}
 1 &  \sigma_{i12} &  \sigma_{i13}\\
 \sigma_{i12} & \sigma_{i2}^2 &  \sigma_{i23}\\
 \sigma_{i13} &  \sigma_{i23} & \sigma_{i3}^2
 \end{pmatrix},
 \]
 with
 \begin{eqnarray*}
 \sigma_{i2}^2 &=& {\lambda_{21}(x_i)^2+1}\\
 \sigma_{i3}^2 &=& {(\lambda_{21}(x_i)\lambda_{32}(x_i)-\lambda_{31}(x_i))^2+\lambda_{32}(x_i)^2+1}\\
 \sigma_{i12} &=& -\lambda_{21}(x_i),\\
 \sigma_{i13} &=& \lambda_{21}(x_i)\lambda_{32}(x_i)-\lambda_{31}(x_i),\\
 \sigma_{i23} &=& -\lambda_{21}(x_i)^2 \lambda_{32}(x_i)+\lambda_{21}(x_i)\lambda_{31}(x_i)-\lambda_{32}(x_i).
 \end{eqnarray*}
To compute $\yvec_{ir}$ in Step~3, we additionally chose $F_3$ to be the CDF of another Dagum distribution with parameters $a_3=\exp(1.5)$, $b_3=\exp(-0.9)$ and $p_3=\exp(1)$, such that
\[
\yvec_{ir}=[F_{1}^{-1}\{\Phi_{0,1}(\tilde z_{ir,1})\},F_{2}^{-1}\{\Phi_{0,\sigma_{i2}^2}(\tilde z_{ir,2})\},F_{3}^{-1}\{\Phi_{0,\sigma_{i3}^2}(\tilde z_{ir,3})\}]^\top,\quad i=1,\ldots,n;\,r=1,\ldots,R.
\]
Note that the marginals of $y_2$ and $y_3$ (or more precisely
their marginal parameters $b_2$ and $b_3$) depend on the covariate $x$.


\paragraph{Results.}
Figure~\ref{fig:sim3d} shows the function estimates for
the three parameters $\lambda_{21}(x_i)$, $\lambda_{31}(x_i)$ and
$\lambda_{32}(x_i)$.
The grey lines indicate overall convincing results for all replicates,
without any problematic outliers even though the estimation errors increased
with the increasing complexity of the functional form. We omit the RMSE plot because it yielded qualitatively  the same results.

\begin{figure}[h]
\centering\includegraphics[width=0.99\textwidth,angle=0]{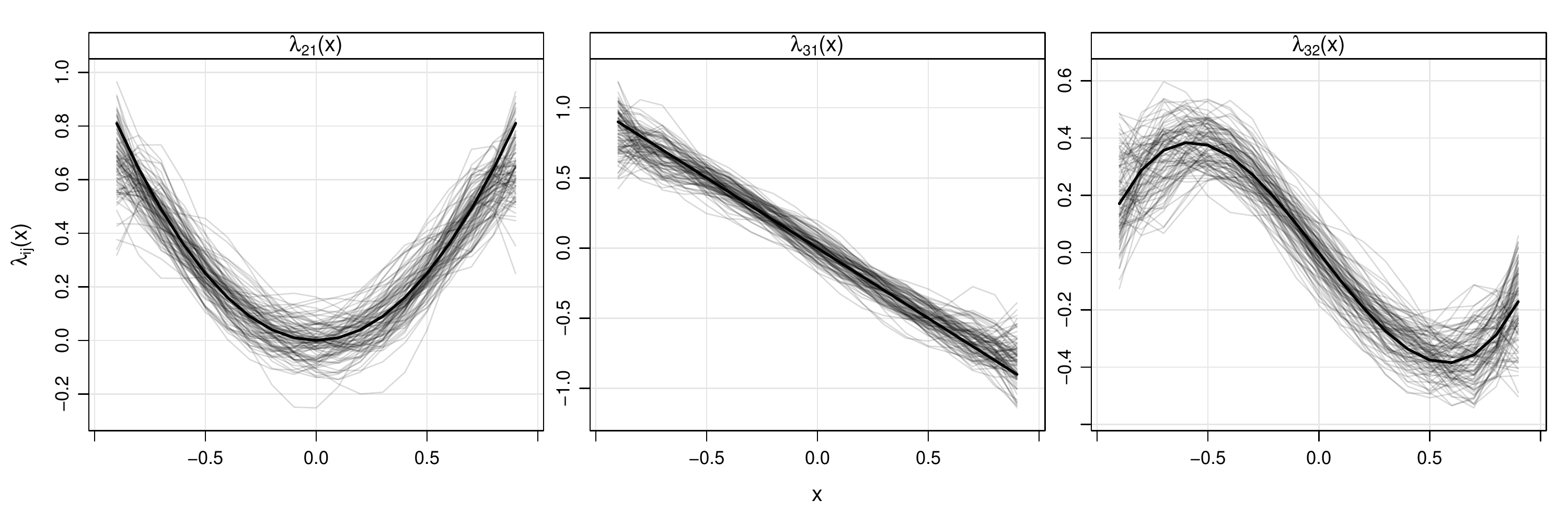}
\caption{ Trivariate simulation. Function estimates of $\lambda_{21}(x)$ (left), $\lambda_{31}(x)$ (middle) and $\lambda_{32}(x)$ (right). The black line is the true function and the grey lines are the estimates of $R=100$ replicates.}\label{fig:sim3d}
\end{figure}

\subsection{Higher-Dimensional Responses}\label{sim3}

{To investigate the performance of MCTMs in higher response dimensions, we
conducted an experiment with $5$- and $10$-dimensional responses.  Specifically,
we employ the settings from Section~\ref{sim2} but add $2$ and $7$ marginally Dagum
distributed components assuming independence, \ie $\lambda_{ij}(x)=0$ for
$i>j,i>3$.  The function estimates for the three
parameters $\lambda_{21}(x_i)$, $\lambda_{31}(x_i)$ and $\lambda_{32}(x_i)$
is qualitatively similar to the results in Figure~\ref{fig:sim3d}, see Figure~\ref{fig:sim510d},
while the true zero components of $\Lambda$ are identified correctly (not
shown in the Figure).  Furthermore the scale of the RMSE does not increase
but is for all replicates and all $\lambda_{ij}(x)<0.2$ as before, so that we omit the additional plot. Overall, these results indicate  very satisfying results even for high-dimensional response situations.}

\begin{figure}[h]
\centering\includegraphics[width=0.99\textwidth,angle=0]{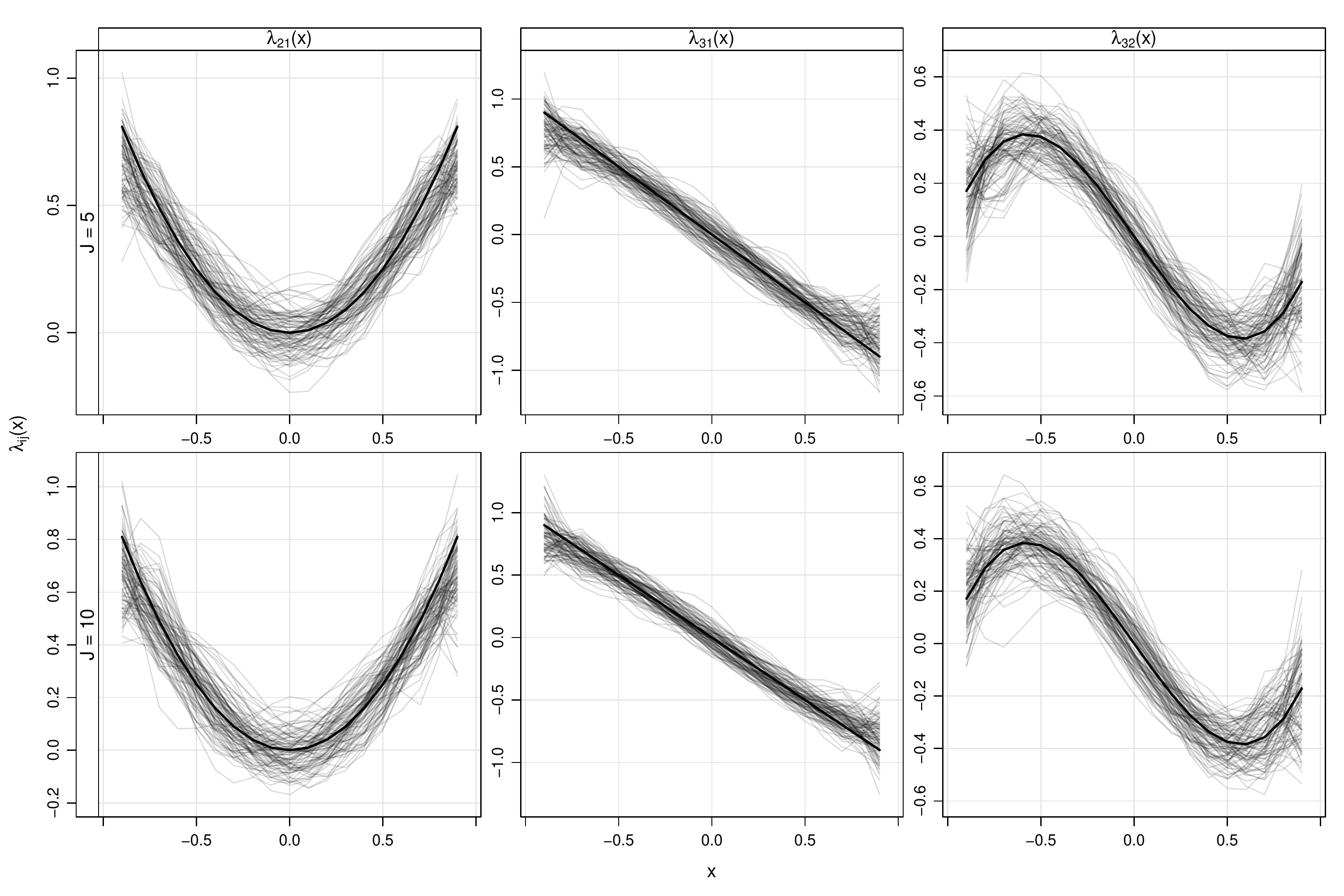}
\caption{ 5-dimensional simulation (upper row) and 10-dimensional simulation (lower row). Each column shows function estimates of $\lambda_{21}(x)$ (left), $\lambda_{31}(x)$ (middle) and $\lambda_{32}(x)$ (right). The black line is the true function and the grey lines are the estimates of $R=100$ replicates.}\label{fig:sim510d}
\end{figure}

\section{Summary and Discussion}\label{sec:conc}

Renewed interest in transformation models
\citep{manuguerra_heller_2010,mclain_ghosh_2013,chernozhukov_2013,
HotKneBue2014,Liu_Shepherd_Li_2017,moehotbue2017,garcia_marder_wang_2018} has been
motivated by the combination of model flexibility, parameter
interpretability, and broad applicability of this class of regression
models.  Rather than assuming a specific distribution of the response,
transformation models rely on a suitable transformation
of the response into an a priori defined reference distribution.
The problem of directly estimating a distribution is replaced by the problem
of estimating this transformation function. However, in many cases,
conceptually and computationally simple solutions to this problem exist
\citep{Hothorn_2018_JSS}.

The MCTMs introduced herein apply
this core principle to multivariate regression.  Similar technical approaches
have been used in discriminant analysis \citep[][refer to transnormal
models]{Lin_Jeon_2003}, quantile regression \citep{Fan_Xue_Zou_2016},
receiver-operating characteristic curve analysis \citep{LyuYinZha2019}, and
are ubiquitous in neural networks as flows,
but the generality of multivariate transformation models for regression
purposes had yet to be fully developed.  MCTMs enjoy the
same flexibility, parameter interpretability and broad applicability as
their univariate counterparts.  The models are highly adaptive and, in our
simulation experiments, performed akin to parametric models that exactly matched
the data-generating process. While the parametric models are often restricted to bivariate responses, our MCTMs work well far beyond that as illustrated empirically for up to ten dimensions.  Appropriate model parameterisations allow both
the marginal distributions and the joint distribution to depend on
covariates.  An important application of this new model class is the
estimation of conditional dependencies, while accounting for covariate
effects in the marginal distributions.

{Conceptually, our framework  carries over to multivariate random
vectors that are discrete or censored.  In particular, both discrete and
censored data can be interpreted as incomplete information $\ubar{\yvec}_i <
\yvec_i \le \bar{\yvec}_i$, where, instead of exact observations $\yvec_i
\in \dsR^J$, only the upper and lower boundaries $\ubar{\yvec}_i$ and
$\bar{\yvec}_i$ respectively, are observed.  For discrete data, the underlying rationale
would be that discrete realisations are obtained by discretisation from
an underlying continuous process.  For censored data, the interval
boundaries result from the censoring mechanism, and random right censoring, left
censoring as well as interval censoring can be handled by appropriate
choices of $\ubar{\yvec}_i$ and $\bar{\yvec}_i$.
\\
For $\Prob_\rZ = \ND(0, 1)$, the log-likelihood contributions are then given by
\[
\ell_i(\thetavec) = \log\left(
\int_{\tilde{\h}(\ubar{\yvec}_i)}^{\tilde{\h}(\bar{\yvec}_i)}
\phi_{\nullvec, \mSigma}(\tilde{\zvec}) \, d\tilde{\zvec}
\right),
\]
where $\tilde{\h}(\yvec) = (\tilde{\h}_1(\ry_1), \dots,
\tilde{\h}_J(\ry_J))$ and $\mSigma=\mLambda^{-1}\mLambda^{-\top}$.
Numerical approximations need to be applied in evaluating these
likelihood contributions. For $J > 2$, the quasi-Monte-Carlo algorithm by
\cite{Genz1992} seems especially appropriate because it relies on the
Cholesky factor $\mLambda$ of the precision matrix rather than the
covariance matrix $\mSigma$.  Nonetheless, the simplicity and explicit structure
of the score contributions from the previous sections do not carry over to
these more general cases, and the numerical evaluation of the log-likelihood
becomes more demanding. We will investigate these challenges in some future work.}

 More complex models, for example, models
featuring additive or spatial effects are conceptually easy to integrate if present in the data
because they only require the addition of suitable
penalty terms to the log-likelihood.  In addition, the analytic
expressions for score and Fisher information functions presented herein
apply only to a standard normal reference distribution; adaptations to the
general case beyond linear dependence structures are still needed.

\section*{Computational Details}

\NK{A reference implementation of conditional and unconditional multivariate
transformation models is available in package \textbf{tram}
\citep{pkg:tram}.  Augmented Lagrangian Minimization implemented in the
\texttt{auglag()} function of package \textbf{alabama} \citep{pkg:alabama} was
used for optimising the log-likelihood, with starting values obtained from
marginal transformation models.

Source code for the reproduction of the empirical results presented in
Sections~\ref{sec:mctm} and \ref{sec:sim} is distributed as part of this
package; the two illustrations can be executed
from within \textsf{R}}
\begin{verbatim*}
install.packages("tram")
library("tram")
example(mmlt)
demo("undernutrition")
\end{verbatim*}
Empirical results were obtained using \textsf{R} \citep[version 4.0.2.,][]{R},
a developer version of BayesX \citep[][]{BelBreKleKneLanUml2015},
\textbf{tram} \citep[version 0.5-1,][]{pkg:tram}, and
\textbf{VGAM} \citep[version 1.1-3,][]{pkg:VGAM}. Source code for
simulations is available from
\begin{verbatim*}
system.file("simulations",package="tram")
\end{verbatim*}

\section*{Acknowledgements}
The authors thank two referees and an associate editor for helpful comments that
improved the manuscript. The authors gratefully acknowledge funding through the Emmy
Noether grant KL 3037/1-1 (Nadja Klein) from the Deutsche Forschungsgemeinschaft (DFG, German Research Foundation), and SNF grant 200021-184603 from the Swiss National Science Foundation (Torsten Hothorn).

\appendix

\renewcommand{\thesection}{Part~\Alph{section}}
\setcounter{figure}{0} \renewcommand\thefigure{\Alph{figure}}
\renewcommand\thetable{\Alph{table}}

\section{Some Properties of the MCTM}\label{app:theos}

\subsection{Dependence Structure}\label{app:theos:dep}

Let  $\mP=\mSigma^{-1} = \mLambda^\top \mLambda$ be the precision matrix of the distribution of
$\mtildeZ$, and $\mY_{-j\jmath}$ and $\mtildeZ_{-j\jmath}$ denote the vectors of the
observed and transformed responses, excluding elements $j$ and $\jmath$. Let furthermore $\mR= \mS\mSigma \mS$, $\mS=\diag(\sigma_1^{-1},\ldots,\sigma_J^{-1})$, be the corresponding correlation matrix and $C_{\scriptsize{Ga}}^{j\jmath}$ the bivariate sub-Gaussian copula of components $j$ and $\jmath$ and correlation matrix entry $R[j,\jmath]$.

\textbf{Conditional Independence. }{\em The entries in $\mLambda$ determine the conditional
independence structure between the transformed responses $\tilde{Z}_j$ (and
therefore, implicitly, also the observed responses $Y_j$), i.e.}
\[
 Y_j \perp Y_\jmath \mid \mY_{-j\jmath}\quad\Longleftrightarrow\quad
\tilde{\rZ}_j \perp
\tilde{\rZ}_\jmath \mid \mtildeZ_{-j\jmath}\quad\Longleftrightarrow\quad \mP[j,\jmath]=0
\]
\begin{proof}
Since $\dsP(\mY\le \yvec) = \dsP(\mtildeZ\le\ztildevec)$ holds, the dependence structure of $\mY$ is that of $\tilde\mZ$. However, $\tilde\mZ~\sim\ND_J(\nullvec_J,\Sigma)$ with $\mSigma = \mLambda^{-1} \mLambda^{-\top}$. Hence, the result is a direct consequence of Theorem 2.2 of  \citet{RueHel2005}
\end{proof}

\textbf{Measures of Dependence.}
\emph{\begin{itemize}
\item[(i)]
For $q\in(0,1)$ and $U_{j}= F_j(\tilde h_{j}^{-1}(\tilde Z_{j}))$, $U_{\jmath}= F_j(\tilde h_{\jmath}^{-1}(\tilde Z_{\jmath}))$, the lower and upper quantile dependence are
\begin{eqnarray*}
\lambda^L_{j\jmath}(q|\mLambda) &\equiv &\mbox{Pr}(U_j<q|U_{\jmath}<q)=\frac{C_{\scriptsize{Ga}}^{j\jmath}(q,q)}{q}\,,\mbox{ and}\\
\lambda^U_{j\jmath}(q|\mLambda) &\equiv &\mbox{Pr}(U_j>q|U_{\jmath}>q)=\frac{1-2q+C_{\scriptsize{Ga}}^{j\jmath}(q,q)}{1-q}\,.
\end{eqnarray*}
\item[(ii)] The lower and upper extremal tail dependence
\[
\lambda_{j \jmath}^L=\lim_{q \downarrow 0}\lambda^L_{j\jmath}(q|\mLambda)=0\,,\mbox{ and }
\lambda_{j \jmath}^U=\lim_{q\uparrow 1}\lambda^U_{j\jmath}(q|\mLambda)=0\,.
\]
\item[(iii)] Spearman's rho is
\[
\rho^S_{j \jmath}(\mLambda) =
 \frac{6}{\pi} \arcsin(R[j,\jmath]/2)\,,
\]
where $R[j,\jmath]$ is as defined above and is a function of
$\mLambda$.
\item[(iv)] Kendall's tau is
\[
\tau^K_{j \jmath}(\mLambda) =
\frac{2}{\pi} \arcsin(R[j,\jmath])\,.
\]
\end{itemize}
}
\begin{proof}
The proofs for (i) and (ii) can be found in \cite{ColHefTaw1999}, and for (iii) and (iv) in \cite{FanFanKot2002}.
\end{proof}

\subsection{Details on Marginal and Conditional Distributions}\label{app:marg:cond}

Assume that the random
vector $\mY$ is partitioned into $\mY_{\calI}$ and $\mY_{\calI^\calC}$, where
$\calI\subset\{1,\ldots,J\}$ is a non-empty set of $I$ indices
$j_1,\ldots,j_I$ and $\calI^{\calC}=\{1,\ldots,J\}\setminus\calI$ is its
complement, consisting of all indices not contained in $\calI$.
The vectors $\mtildeZ_{\calI}$ and $\mtildeZ_{\calI^\calC}$ can be similarly defined.  The
marginal distribution of $\mtildeZ_{\calI}$ is then given by
$\ND_I(\nullvec_I,\mSigma_\calI)$, where $\mSigma_\calI$ is the submatrix of
$\mSigma$ containing the elements related to the subset $\calI$.  We can
therefore deduce both the marginal CDF and the density of
$\mY_\calI$ as
\[
 \dsP(\mY_\calI\le \yvec_\calI) = \dsP(\mtildeZ_\calI\le\ztildevec_\calI) = \Phi_{\nullvec_I, \mSigma_\calI}(\ztildevec_\calI)
\]
where $\tilde{z}_{j_i}=\tilde{\h}_{j_i}(\ry_{j_i})$ and
\begin{equation}\label{eq:margdens}
f_{\mY_\calI}(\yvec_\calI) = \phi_{\nullvec_I, \mSigma_\calI}(\ztildevec_\calI)\prod_{j_i\in\calI} \frac{\partial
\tilde{\h}_{j_i}(\ry_{j_i})}{\partial \ry_{j_i}}.
\end{equation}
We use a similar process for the conditional distribution of $\mY_\calI$, given $\mY_{\calI^\calC} =
\yvec_{\calI^\calC}$, and note that
\[
 \dsP(\mY_\calI\le\yvec_\calI \mid \mY_{\calI^\calC} = \yvec_{\calI^\calC}) =
 \dsP(\mtildeZ_\calI\le\ztildevec_\calI \mid \mY_{\calI^\calC} = \yvec_{\calI^\calC}) =
 \dsP(\mtildeZ_\calI\le\ztildevec_\calI \mid \mtildeZ_{\calI^\calC} =  \tilde{\zvec}_{\calI^\calC}).
\]
From the rules for multivariate normal distributions we furthermore obtain
\[
 \mtildeZ_\calI \mid \mtildeZ_{\calI^\calC} =  \tilde{\zvec}_{\calI^\calC} \sim\ND_I(\muvec_{\calI \mid \calI^\calC}, \mSigma_{\calI \mid \calI^\calC})
\]
with
\[
 \muvec_{\calI \mid \calI^\calC} = \mSigma_{\calI,\calI^\calC}\mSigma_{\calI^\calC}^{-1} \tilde{\zvec}_{\calI^\calC} \text{ and }  \mSigma_{\calI \mid \calI^\calC} = \mSigma_\calI -
\mSigma_{\calI,\calI^\calC}\mSigma_{\calI^\calC}^{-1}\mSigma_{\calI,\calI^\calC}^\top
\]
where $\mSigma_{\calI^\calC}$ and $\mSigma_{\calI,\calI^\calC}$ denote
the sub-blocks of $\mSigma$ corresponding to the respective index sets.
Therefore, the conditional density of  $\mY_\calI$ given $\mY_{\calI^\calC}$
is
\[
 f_{\mY_{\calI} \mid \mY_{\calI^\calC}}(\yvec_\calI) = \phi_{\muvec_{\calI
\mid \calI^\calC}, \mSigma_{\calI \mid \calI^\calC}}(\ztildevec_\calI)\prod_{j_i\in\calI} \frac{\partial
\tilde{\h}_{j_i}(\ry_{j_i})}{\partial \ry_{j_i}}.
\]

\section{Observed Fisher Information}\label{app:fisher}

Let $\mathcal{F}_i(\thetavec)=-\frac{\partial^2
l_i(\thetavec)}{\partial\thetavec\partial\thetavec^\top}=-\frac{\partial^2
l_i(\thetavec)}{\partial(\parm^\top,\lambdavec^\top)^\top\partial(\parm^\top,\lambdavec^\top)}$. Then the elements of the observed Fisher information are given by
\begin{equation}
\begin{aligned}
\mathcal{F}_i(\lambda_{\tilde k k},\lambda_{\tilde l l})&= -\frac{\partial^2 l_i(\thetavec)}{\partial \lambda_{\tilde k k}\partial\lambda_{\tilde l l}}\\
&=
\begin{cases}
\basisy_k(\ry_{ik})^\top \parm_k \basisy_l(\ry_{il})^\top \parm_l & \mbox{if } \tilde k=\tilde l \mbox{ and } l<\tilde k\\
0 & \mbox{otherwise},
\end{cases}
\end{aligned}
\end{equation}

\begin{equation}
\begin{aligned}
\mathcal{F}_i(\lambda_{\tilde k k},\parm_l)&= -\frac{\partial^2 l_i(\thetavec)}{\partial \lambda_{\tilde k
k}\partial\parm_l^\top}\\
&=
\begin{cases}
\lambda_{\tilde k l}\basisy_l(y_{il})\basisy_k(\ry_{ik})^\top \parm_k &\mbox{if } l<\tilde k, l\neq k\\
 \basisy_{\tilde k}(\ry_{i\tilde k})  \basisy_{k}(\ry_{ik})^\top \parm_k & \mbox{if } \tilde k=l\\
  \left(\sum_{\jmath = 1}^{\tilde k - 1}
  \lambda_{\tilde k\jmath} \basisy_\jmath(\ry_{i\jmath})^\top \parm_\jmath +
  \basisy_{\tilde k}(\ry_{i\tilde k})^\top \parm_{\tilde k} \right) \basisy_k(\ry_{ik})\\
    \quad  +  \lambda_{\tilde k k}\basisy_k(\ry_{ik})\basisy_k(\ry_{ik})^\top \parm_k& \mbox{if } k=l\\
0 & \mbox{otherwise,}
\end{cases}
\end{aligned}
\end{equation}

and

\begin{equation}
\begin{aligned}
\mathcal{F}_i(\parm_{\tilde k},\parm_k)&= -\frac{\partial^2 l_i(\thetavec))}{\partial
\parm_{\tilde k}\partial\parm_k^\top}\\
&=
\begin{cases}
\sum_{j=k}^J\lambda_{jk}^2\basisy_k(\ry_{ik})\basisy_k(\ry_{ik})^\top +
 \frac{\basisy^\prime_k(\ry_{ik})\basisy^\prime_k(\ry_{ik})^\top}{(\basisy^\prime_k(\ry_{ik})^\top \parm_k)^2}&\mbox{if }   k=\tilde k\\
\sum_{j = \tilde k}^J \lambda_{jk}\basisy^\prime_k(\ry_{ik}) \lambda_{j\tilde k}\basisy^\prime_{\tilde k}(\ry_{i\tilde k})^\top
\end{cases}
\end{aligned}
\end{equation}

\section{Fast Alternative to Parametric Boostrap}\label{app:undernutrition}

\NK{As discussed in Sectio~\ref{sec:pb}, a fast alternative to the parametric bootstrap procedure for cases in which $n$ is large is to draw samples $\hat\thetavec_{(b)}$, $b=1,\ldots,B$ from the asymptotic normal distribution of $\thetavec$ from Corollary~\ref{cor2}. We here present the Figure~\ref{fig:undernutrition:both} showing the $95\%$ confidence intervals of the  the conditional rank
correlations $\rho^S$ between \text{stunting}, \text{wasting} and \text{underweight}
as functions of $\text{age}$ of the parametric bootstrap and the fast alternative.}

\begin{figure}[t]
\centering\includegraphics[width=0.99\textwidth,angle=0]{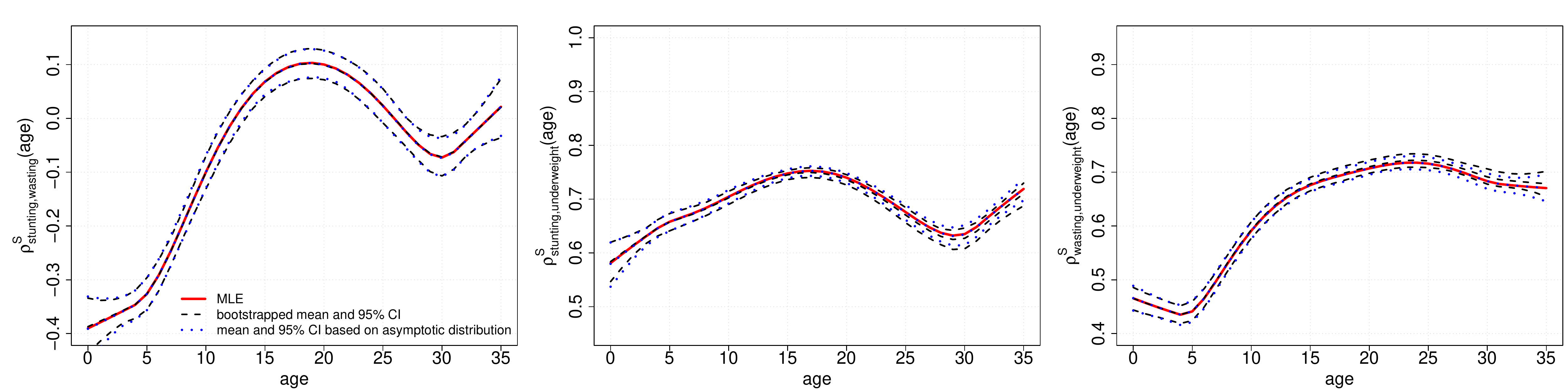}
\caption{Undernutrition.  Spearman's rho
$\rho_{12}^S(\text{age})=\rho^S(y_{\text{\scriptsize{stunting}}},y_{\text{\scriptsize{wasting}}}|\text{age})$,
$\rho_{13}^S(\text{age})=\rho^S(y_{\text{\scriptsize{stunting}}},y_{\text{\scriptsize{underweight}}}|\text{age})$
and
$\rho_{23}^S(\text{age})=\rho^S(y_{\text{\scriptsize{wasting}}},y_{\text{\scriptsize{underweight}}}|\text{age})$.
Shown are the maximum likelihood estimates (solid red line),
the bootstrapped mean estimate and the 95\% bootstrapped
confidence intervals (dashed black lines), as well as mean estimate and 95\% confidence intervals (dashed blue lines) based on the asymptotic distribution of $\thetavec$.} \label{fig:undernutrition:both}
\end{figure}

%
%
%

\small\bibliography{litliste}

\end{document}